\documentclass{llncs}
\usepackage{amssymb, amsfonts, amsmath}
\usepackage{breakurl}             
\usepackage{underscore}           
\usepackage{thmtools}

\usepackage{graphicx}
\graphicspath{ {./handwritten/cropped/handwritten-} }

\usepackage[svgnames]{xcolor}

\definecolor[named]{rpcolor}{cmyk}{0.80,0,1,0.19} 
\definecolor[named]{gbcolor}{cmyk}{1,0.1,0,0.2} 
\definecolor[named]{asidecolor}{gray}{0.7} 

\newcommand{\Syn}{\mathbf{Syn}}
\newcommand{\Sem}{\mathbf{Sem}}
\newcommand{\IH}{\mathsf{IH}}

\newcommand{\IHaff}{\IH^+}

\newcommand{\IHpoly}{\IH^+_{\geq}}
\newcommand{\IHpl}{\IH_{PL}}
\newcommand{\eqE}[1]{\mathrel{\overset{\makebox[0pt]{\mbox{\normalfont\tiny #1}}}{=}}}
\newcommand{\leqE}[1]{\mathrel{\overset{\makebox[0pt]{\mbox{\normalfont\tiny #1}}}{\subseteq}}}

\newcommand{\polyleq}{\ \leqE{$\IHpoly$}\ }
\newcommand{\pleq}{\ \eqE{$\IHpl$}\ }
\newcommand{\plleq}{\ \leqE{$\IHpl$}\ }
\newcommand{\Hyperplanes}{\mathsf{Hyperplanes}}
\newcommand{\Interior}{\mathsf{Int}}
\newcommand{\freeProp}{\mathsf{P}}
\newcommand{\freeUProp}{\mathsf{P}_\cup}

\usepackage{tikz}
\usetikzlibrary{arrows,backgrounds,shapes.geometric,shapes.misc,matrix,arrows.meta,decorations.markings}

\pgfdeclarelayer{background}
\pgfdeclarelayer{edgelayer}
\pgfdeclarelayer{nodelayer}
\pgfsetlayers{background,edgelayer,nodelayer,main}
\tikzset{baseline=-0.5ex}
\definecolor{light-gray}{gray}{.7}
\tikzstyle{none}=[inner sep=0pt]
\tikzstyle{plain}=[inner sep=0pt]
\tikzstyle{black}=[circle, draw=black, fill=black, inner sep=0pt, minimum size=3.5pt]
\tikzstyle{black-faded}=[circle, draw=light-gray, fill=light-gray, inner sep=0pt, minimum size=4pt]
\tikzstyle{white}=[circle, draw=black, fill=white, inner sep=0pt, minimum size=3.5pt]
\tikzstyle{white-faded}=[circle, draw=light-gray, fill=white, inner sep=0pt, minimum size=4.5pt]
\tikzstyle{delay}=[fill=black, regular polygon, regular polygon sides=3,rotate=-90, scale=.55]
\tikzstyle{delay-op}=[fill=black, regular polygon, regular polygon sides=3,rotate=90, scale=.55]
\tikzstyle{reg}=[draw, fill=white, rounded rectangle, rounded rectangle left arc=none, minimum height=1em, minimum width=1em, node font={\scriptsize}]
\tikzstyle{coreg}=[draw, fill=white, rounded rectangle, rounded rectangle right arc=none, minimum height=1em, minimum width=1em, node font={\scriptsize}]
\tikzstyle{rn}=[circle, draw=red, fill=red, inner sep=0pt, minimum size=4pt]
\tikzstyle{place}=[circle, draw=black, fill=white, inner sep=0pt, minimum size=8pt]

\tikzstyle{medium box}=[fill=white, draw=black, shape=rectangle, minimum height=1cm, minimum width=0.75cm]
\tikzstyle{small box}=[fill=white, draw=black, shape=rectangle, minimum height=0.75cm, minimum width=0.5cm]
\tikzstyle{square}=[fill=white, draw=black, shape=rectangle, minimum height=1cm, minimum width=1cm]
\tikzstyle{triangle}=[fill=black, draw=black, shape=regular polygon, regular polygon sides=3, rotate=270, scale=0.6]
\tikzstyle{antipode}=[fill=black, draw=black, shape=rectangle]
\tikzstyle{triangleop}=[fill=black, draw=black, shape=regular polygon, regular polygon sides=3, rotate=90, scale=0.6]
\tikzstyle{label}=[fill=none, draw=none, shape=circle]
\tikzstyle{transition}=[fill=white, draw=black, shape=rectangle, minimum height=0.75cm, minimum width=0.75cm]

\tikzstyle{transition}=[rectangle,thick,draw=black!75,fill=black!20,minimum size=7pt]
\tikzstyle{place}=[circle, draw=black, fill=white, inner sep=0pt, minimum size=8pt]
\tikzstyle{arrow}=[->]

\newcommand{\One}{
\tikzset{x=1em, y=2.1ex}
\begin{tikzpicture}[baseline=-.5ex]
	\begin{pgfonlayer}{nodelayer}
		\node [style=none] (0) at (-0.75, -0) {};
		\node [style=none] (1) at (0.25, -0) {};
		\node [style=none] (2) at (-0.75, 0.25) {};
		\node [style=none] (3) at (-0.75, -0.25) {};
	\end{pgfonlayer}
	\begin{pgfonlayer}{edgelayer}
		\draw (0.center) to (1.center);
		\draw (3.center) to (2.center);
	\end{pgfonlayer}
\end{tikzpicture}}
\tikzset{x=1em, y=1.5ex}
}

\usepackage[T1]{fontenc}
\usepackage[utf8]{inputenc}
\usepackage[mathscr]{eucal}

\usepackage{stmaryrd}
\usepackage{hyperref}
\usepackage[capitalise]{cleveref}
\usepackage{tikz}
\usepackage{tikz-cd}

\usetikzlibrary{arrows,snakes,automata,backgrounds,petri, positioning,calc,cd, automata}
\usetikzlibrary{shapes.geometric,shapes.misc,matrix,arrows.meta,decorations.markings}
\tikzset{x=1em, y=1.5ex, baseline=-0.5ex}
\tikzset{ihbase/.style={inner sep=0,circle,draw,fill=lightgray,minimum size=0.4em,node contents={}}}
\tikzset{ihblack/.style={ihbase,fill=black}}
\tikzset{ihwhite/.style={ihbase,fill=white}}
\tikzset{mat/.style={draw,fill=white,rectangle,node font=\scriptsize}}
\tikzset{ha/.style={mat,rounded rectangle,rounded rectangle left arc=none}}
\tikzset{haop/.style={mat,rounded rectangle,rounded rectangle right arc=none}}
\tikzset{blackha/.style={mat,rounded rectangle,rounded rectangle left arc=none,font=\color{white},fill=black}}
\tikzset{blackhaop/.style={mat,rounded rectangle,rounded rectangle right arc=none,font=\color{white},fill=black}}
\tikzset{anti/.style={inner sep=0,isosceles triangle,fill=black,draw=black, minimum width=0.75em, node contents={}}}
\tikzset{antiop/.style={anti,shape border rotate=180}}
\tikzset{antisq/.style={inner sep=0,rectangle,fill=black, minimum height=1em, minimum width=0.6em, node contents={}}}
\tikzset{count/.style={above,inner ysep=0.15em,font=\scriptsize}}
\tikzset{axiom/.style={above,font=\small}}
\tikzset{dir/.style={-Latex}}
\tikzset{st/.style={decoration={markings,
    mark={at position 0.5 with {\draw (0, 2pt) to (0, -2pt);}}},
    postaction=decorate}}

\newcommand{\sem}[1]{\left\llbracket{\,#1\,}\right\rrbracket}

\newcommand{\genericcomult}[2]{
  \begin{tikzpicture}
    \node at (1, 0) [ihbase,solid,name=copy,#1];
    \draw[#2] (copy) .. controls (1.25, 0.75) .. (2, 0.75);
    \draw[#2] (0, 0) -- (copy);
    \draw[#2] (copy) .. controls (1.25, -0.75) .. (2, -0.75);
  \end{tikzpicture}
}

\newcommand{\genericcounit}[2]{
  \tikz \draw[#2] (0, 0) -- (1, 0) node[ihbase,#1, solid];
}
\newcommand{\genericcounitn}[2]{
  \tikz \draw (0, 0) -- node[count] {#2} (1, 0) node[ihbase,#1];
}
\newcommand{\genericmult}[2]{
  \tikz {
    \node at (1,0) (copy) [ihbase,#1,solid];
    \draw[#2] (0,  0.75) .. controls (0.75,  0.75) .. (copy);
    \draw[#2] (0, -0.75) .. controls (0.75, -0.75) .. (copy);
    \draw[#2] (copy) -- (2, 0);
  }
}

\newcommand{\genericunit}[2]{
  \tikz \draw[#2] (0, 0) node[ihbase,#1, solid] -- (1, 0);
}

\newcommand{\Bcomult}{\genericcomult{ihblack}{}}

\newcommand{\Bcounit}{\genericcounit{ihblack}{}}

\newcommand{\Bmult}{\genericmult{ihblack}{}}

\newcommand{\Bunit}{\genericunit{ihblack}{}}

\newcommand{\Wmult}{\genericmult{ihwhite}{}}

\newcommand{\Wunit}{\genericunit{ihwhite}{}}

\newcommand{\Wcomult}{\genericcomult{ihwhite}{}}

\newcommand{\Wcounit}{\genericcounit{ihwhite}{}}
\newcommand{\Wcounitn}[1]{\genericcounitn{ihwhite}}

\newcommand\idzero{
\tikzset{x=1em, y=2.1ex}
\InputIfFileExists{empty-diag.tikz}{}{\input{./tikz/empty-diag.tikz}}
\tikzset{x=1em, y=1.5ex}
} 

\newcommand{\idone}{
  \tikz \draw (0, 0) -- (1, 0);
}

\newcommand{\sym}{
  \tikz {
    \draw (0,  0.5) .. controls (0.5,  0.5) and (0.5, -0.5) .. (1, -0.5);
    \draw (0, -0.5) .. controls (0.5, -0.5) and (0.5,  0.5) .. (1,  0.5);
  }
}

\newcommand\scalar[1]{
  \tikz {
    \node[ha] (ha) {$#1$};
    \draw (ha.west) -- ++(-0.75, 0);
    \draw (ha.east) -- ++(0.75, 0);
  }
}

\newcommand\coscalar[1]{
  \tikz {
    \node[haop] (haop) {$#1$};
    \draw (haop.west) -- ++(-0.75, 0);
    \draw (haop.east) -- ++(0.75, 0);
  }
}

\tikzset{x=1em, y=1.5ex}
}

\tikzset{x=1em, y=1.5ex}
}

\pgfdeclarelayer{edgelayer}
\pgfdeclarelayer{nodelayer}
\pgfsetlayers{edgelayer,nodelayer,main}

\definecolor{light-gray}{gray}{.5}
\tikzstyle{none}=[inner sep=0pt]
\tikzstyle{plain}=[inner sep=0pt]
\tikzstyle{black}=[circle, draw=black, fill=black, inner sep=0pt, minimum size=4pt]
\tikzstyle{black-faded}=[circle, draw=light-gray, fill=light-gray, inner sep=0pt, minimum size=4pt]
\tikzstyle{white}=[circle, draw=black, fill=white, inner sep=1pt, minimum size=4.5pt, font={\scriptsize}]
\tikzstyle{white-faded}=[circle, draw=light-gray, fill=white, inner sep=0pt, minimum size=4.5pt]
\tikzstyle{delay}=[fill=black, regular polygon, regular polygon sides=3,rotate=-90, scale=.55]
\tikzstyle{delay-op}=[fill=black, regular polygon, regular polygon sides=3,rotate=90, scale=.55]
\tikzstyle{reg}=[draw, fill=white, rounded rectangle, rounded rectangle left arc=none, minimum height=1.2em, minimum width=1.4em, inner sep=1pt, node font={\scriptsize}]
\tikzstyle{coreg}=[draw, fill=white, rounded rectangle, rounded rectangle right arc=none, minimum height=1.2em, minimum width=1.4em, node font={\scriptsize}]
\tikzstyle{basic rounded box}=[draw, rectangle, rounded corners, minimum height=1.2em, minimum width=1.4em]
\tikzstyle{small rounded box}=[draw, rectangle, rounded corners, minimum height=1.2em, minimum width=1.4em, node font={\scriptsize}]
\tikzstyle{rcoreg}=[draw=red, fill=white, rounded rectangle, rounded rectangle right arc=none, minimum height=1.2em, minimum width=1.4em, node font={\scriptsize}]
\tikzstyle{regb}=[draw, fill=black, rounded rectangle, rounded rectangle left arc=none, minimum height=1.2em, minimum width=1.4em, node font={\scriptsize}]
\tikzstyle{regbw}=[draw, left color=black, right color=white, middle color=white, rounded rectangle, rounded rectangle left arc=none, minimum height=1.2em, minimum width=1.4em, node font={\scriptsize}]
\tikzstyle{regwb}=[draw, left color=white, right color=black, middle color=white, rounded rectangle, rounded rectangle left arc=none, minimum height=1.2em, minimum width=1.4em, node font={\scriptsize}]
\tikzstyle{coregb}=[draw, fill=black, rounded rectangle, rounded rectangle right arc=none, minimum height=1.2em, minimum width=1.4em, node font={\scriptsize}]
\tikzstyle{coregbw}=[draw, left color=black, right color=white, middle color=white, rounded rectangle, rounded rectangle right arc=none, minimum height=1.2em, minimum width=1.4em, node font={\scriptsize}]
\tikzstyle{coregwb}=[draw, left color=white, right color=black, middle color=white, rounded rectangle, rounded rectangle right arc=none, minimum height=1.2em, minimum width=1.4em, node font={\scriptsize}]
\tikzstyle{rn}=[circle, draw=red, fill=red, inner sep=0pt, minimum size=4pt]
\tikzstyle{wrn}=[circle, draw=red, fill=white, inner sep=0pt, minimum size=4pt]
\tikzstyle{place}=[circle, draw=black, fill=white, inner sep=1pt, minimum size=9pt, node font={\scriptsize}]
\tikzstyle{act}=[circle, draw=black, fill=white, inner sep=0pt, minimum size=4.5pt]
\tikzstyle{coact}=[draw, fill=white, rounded rectangle, rounded rectangle right arc=none, minimum height=.7em, minimum width=.9em, node font={\scriptsize}]
\tikzstyle{triangle right}=[regular polygon, regular polygon sides=3, draw, fill=white, text width=1pt, inner sep=1pt, outer sep=0mm, shape border rotate=-90]
\tikzstyle{triangle left}=[regular polygon, regular polygon sides=3, draw, fill=white, text width=1pt, inner sep=1pt, outer sep=0mm, shape border rotate=90]
\tikzstyle{leq}=[draw, rounded rectangle, rounded rectangle right arc=none, minimum height=1em, minimum width=1em, scale=0.75, inner sep=3pt, font={$\leq$}]
\tikzstyle{geq}=[draw, rounded rectangle, rounded rectangle left arc=none, minimum height=1em, minimum width=1em, scale=0.75, inner sep=3pt, font={$\geq$}]
\tikzstyle{basic rounded box}=[draw, fill=white, rectangle, rounded corners, minimum height=1.2em, minimum width=1.4em]
\tikzstyle{small rounded box}=[draw, fill=white, rectangle, rounded corners, minimum height=1.2em, minimum width=1.4em, node font={\scriptsize}]
\tikzstyle{union}=[draw, fill=white, rounded rectangle, rounded rectangle right arc=none, minimum height=1.2em, minimum width=1.4em, node font={\scriptsize}]
\tikzstyle{basic box}=[draw, fill=white, rectangle, minimum height=1.2em, minimum width=1.4em]
\tikzset{
BWmatrix/.pic={
    \coordinate (center) at (0,0);
    \filldraw[fill=white, draw=black, line width=1pt] (.5,0) 
        [rounded corners=14pt] -- (1,0) 
        [rounded corners=14pt] -- (1,1)
        [rounded corners=0pt] -- (.5,1) 
        [rounded corners=0pt] -- cycle;
    \filldraw[fill=black, draw=black, line width=1pt] (0,0) 
        -- (.5,0) 
        -- (.5,1)
        -- (0,1) 
        -- cycle;
   },
pics/BWmatrix/.default=0.2
}

\tikzstyle{pl}=[circle,thick,draw=black!75,fill=white,minimum size=17pt]
\tikzstyle{port}=[circle, fill,inner sep=1.2pt]
\tikzstyle{transition}=[rectangle,thick,draw=black!75,
  			  fill=black!20,minimum size=7pt]

\tikzstyle{arrow}=[->]

\newcommand{\diagbox}[3]{
\tikzset{x=1em, y=2.1ex}
\begin{tikzpicture}
	\begin{pgfonlayer}{nodelayer}
		\node [style=none] (0) at (-0.75, 0.5) {};
		\node [style=none] (1) at (-0.25, 1) {};
		\node [style=none] (2) at (-0.75, -0.5) {};
		\node [style=none] (3) at (0.75, -0.5) {};
		\node [style=none] (4) at (-0.25, -1) {};
		\node [style=none] (5) at (0.75, 0.5) {};
		\node [style=none] (6) at (2.5, -0) {};
		\node [style=none] (7) at (0.75, -0) {};
		\node [style=none] (8) at (0.25, -1) {};
		\node [style=none] (9) at (0.25, 1) {};
		\node [style=none] (10) at (0, -0) {$#1$};
		\node [style=none] (11) at (2.25, 0.5) {\scriptsize $#3$};
		\node [style=none] (12) at (-2.25, 0.5) {\scriptsize $#2$};
		\node [style=none] (13) at (-2.5, -0) {};
		\node [style=none] (14) at (-0.75, -0) {};
	\end{pgfonlayer}
	\begin{pgfonlayer}{edgelayer}
		\draw [in=180, out=0, looseness=1.25] (7.center) to (6.center);
		\draw [semithick, in=0, out=-90, looseness=1.00] (3.center) to (8.center);
		\draw [semithick, in=-90, out=180, looseness=1.00] (4.center) to (2.center);
		\draw [semithick, in=180, out=90, looseness=1.00] (0.center) to (1.center);
		\draw [semithick, in=90, out=0, looseness=1.00] (9.center) to (5.center);
		\draw [semithick] (1.center) to (9.center);
		\draw [semithick] (5.center) to (3.center);
		\draw [semithick] (8.center) to (4.center);
		\draw [semithick] (2.center) to (0.center);
		\draw [in=180, out=0, looseness=1.25] (13.center) to (14.center);
	\end{pgfonlayer}
\end{tikzpicture}
\tikzset{x=1em, y=1.5ex}
}

\newcommand{\smallubox}[1]{
\tikzset{x=1em, y=2.1ex}
\begin{tikzpicture}
	\begin{pgfonlayer}{nodelayer}
		\node [style=none] (0) at (-0.5, 0.5) {};
		\node [style=none] (1) at (-0.25, 0.75) {};
		\node [style=none] (2) at (-0.5, -0.5) {};
		\node [style=none] (3) at (0.5, -0.5) {};
		\node [style=none] (4) at (-0.25, -0.75) {};
		\node [style=none] (5) at (0.5, 0.5) {};
		\node [style=none] (6) at (1.5, -0) {};
		\node [style=none] (7) at (0.5, -0) {};
		\node [style=none] (8) at (0.25, -0.75) {};
		\node [style=none] (9) at (0.25, 0.75) {};
		\node [style=none] (10) at (0, -0) {$#1$};
		\node [style=none] (11) at (-1.5, -0) {};
		\node [style=none] (12) at (-0.5, -0) {};
	\end{pgfonlayer}
	\begin{pgfonlayer}{edgelayer}
		\draw [in=180, out=0, looseness=1.25] (7.center) to (6.center);
		\draw [semithick, in=0, out=-90, looseness=1.00] (3.center) to (8.center);
		\draw [semithick, in=-90, out=180, looseness=1.00] (4.center) to (2.center);
		\draw [semithick, in=180, out=90, looseness=1.00] (0.center) to (1.center);
		\draw [semithick, in=90, out=0, looseness=1.00] (9.center) to (5.center);
		\draw [semithick] (1.center) to (9.center);
		\draw [semithick] (5.center) to (3.center);
		\draw [semithick] (8.center) to (4.center);
		\draw [semithick] (2.center) to (0.center);
		\draw [in=180, out=0, looseness=1.25] (11.center) to (12.center);
	\end{pgfonlayer}
\end{tikzpicture}
\tikzset{x=1em, y=1.5ex}
}

\newcommand{\myeq}[1]{\mathrel{\overset{\makebox[0pt]{\mbox{\tiny (#1)}}}{=}}}
\newcommand{\myleq}[1]{\mathrel{\overset{\makebox[0pt]{\mbox{\normalfont\tiny (#1)}}}{\subseteq}}}

\newcommand{\R}{\mathbb{R}}

\newcommand{\field}{\mathbb{K}} 

\newcommand{\Mat}[1]{\mathsf{Mat}_{#1}}

\newcommand{\AffRel}[1]{\mathsf{AffRel}_{#1}}

\newcommand{\from}{\mathrel{:}\,}

\newcommand{\poi}{\,;\,}

\newcommand{\adjto}{\,\lower1pt\hbox{$\dashv$}\,}

\newcommand{\LinRel}[1]{\mathsf{LinRel}_{#1}}

\newcommand{\RelX}[1]{\mathsf{Rel}_{#1}}

\newcommand{\PolyRel}[1]{\mathsf{PolyRel}_{#1}}

\newcommand{\df}{:=}

\newcommand{\diagregexp}[1]{
\begin{tikzpicture}
	\begin{pgfonlayer}{nodelayer}
		\node [style=none] (0) at (1.5, 0) {};
		\node [style=rcoreg] (1) at (0, 0) {{\color{red} $e$}};
	\end{pgfonlayer}
	\begin{pgfonlayer}{edgelayer}
		\draw [red] (1) to (0.center);
	\end{pgfonlayer}
\end{tikzpicture}}

\newcommand{\circleeffect}[1]{
\begin{tikzpicture}
	\begin{pgfonlayer}{nodelayer}
		\node [style=none] (3) at (-0.75, 0) {};
		\node [style=place] (5) at (0.75, 0) {$#1$};
	\end{pgfonlayer}
	\begin{pgfonlayer}{edgelayer}
		\draw (3.center) to (5);
	\end{pgfonlayer}
\end{tikzpicture}}

\usepackage{siunitx}
\usepackage[american,siunitx]{circuitikz}
\ctikzset{
    csources/scale=0.5,
    diodes/scale=0.5,
    instruments/scale=0.5,
    resistors/scale=0.5,
    sources/scale=0.5,
    transistors/scale=0.5,
    bipoles/border margin=1,
    bipoles/thickness=1,
    tripoles/thickness=1,
    current arrow scale=24,
}
\tikzstyle{vsource}=[rmeter, t={\textsf{\tiny -- +}}] 
\tikzstyle{ammeter}=[rmeter, t={\textsf{A}}] 
\tikzstyle{vmeter}=[rmeter, t={\textsf{V}}] 
\tikzstyle{elecdot}=[circle,fill,inner sep=0.85pt]

\newcommand{\includegraffle}[1]{
  {\lower10pt\hbox{$\includegraphics[height=1cm]{tikz/#1.pdf}$}}
}
\newcommand{\eleccopy}{\tikz[color=blue]{
  \draw (1,1.5) -- (0,1.5);
  \draw (1,-1.5) -- (0,-1.5);
  \draw (0,1.5) -- (0,-1.5);
  \node[elecdot] at (0,0) {};
  \draw (0,0) -- (-1,0);
}}
\newcommand{\eleccocopy}{\tikz[color=blue]{
  \draw (-1,1.5) -- (0,1.5);
  \draw (-1,-1.5) -- (0,-1.5);
  \draw (0,1.5) -- (0,-1.5);
  \node[elecdot] at (0,0) {};
  \draw (0,0) -- (1,0);
}}
\newcommand{\elecunit}{\tikz[color=blue]{
  \node[elecdot] at (0,0) {};
  \draw (0,0) -- (1,0);
}}
\newcommand{\eleccounit}{\tikz[color=blue]{
  \node[elecdot] at (0,0) {};
  \draw (0,0) -- (-1,0);
}}

\begin{document}

\title{Graphical Piecewise-Linear Algebra}

\author{
    Guillaume Boisseau\inst{1}\orcidID{0000-0001-5244-893X}
    \and Robin Piedeleu\inst{2}
}

\institute{
    University of Oxford, Oxford, UK
    \email{guillaume.boisseau@cs.ox.ac.uk}
    \and
    University College London, London, UK
    \email{r.piedeleu@ucl.ac.uk}
}


\maketitle

\begin{abstract}
    Graphical (Linear) Algebra is a family of diagrammatic languages
    allowing to reason about different kinds
    of subsets of vector spaces compositionally.
    It has been used to model various application domains,
    from signal-flow graphs to Petri nets and electrical circuits.
    In this paper,
    we introduce to the family its most expressive member to date:
    Graphical Piecewise-Linear Algebra, a new language to specify piecewise-linear subsets of vector spaces.

    Like the previous members of the family,
    it comes with a complete axiomatisation,
    which means it can be used to reason about the corresponding semantic domain purely equationally,
    forgetting the set-theoretic interpretation.
    We show completeness using a single axiom on top of Graphical Polyhedral Algebra,
    and show that this extension is the smallest that can capture
    a variety of relevant constructs.

    Finally, we showcase its use by modelling the behaviour
    of stateless electronic circuits of ideal elements, a domain that had remained outside the remit of previous diagrammatic languages. 
    \keywords{string diagrams \and piecewise-linear \and prop \and axiomatisation}
\end{abstract}


\section{Introduction}

Functional thinking underpins most scientific models. Nature, however, does not distinguish inputs and outputs---physical systems are governed by laws that merely express \emph{relations} between their observable variables. While influential scientists, like the famous control theorist J. Willems, have pointed out the inadequacy of functional thinking~\cite{Willems2007}, it has remained the dominant paradigm in science and engineering. Arguably, our mathematical practice, especially the foundational emphasis on sets and functions, and the limitations of standard algebraic syntax, are partially to blame for the persistence of this status quo. But there are also alternative approaches, that take relations seriously as the primitive building blocks of our mathematical models. Category theory in particular is agnostic about what constitutes a morphism and can accommodate relations as easily as functions.

Relations, with their usual composition and the cartesian product of sets, form a monoidal category---a category in which morphisms can be composed in two different ways. As a result, they admit a natural two-dimensional syntax of \emph{string diagrams}. This notation has several advantages when it comes to reasoning about open and interconnected systems~\cite{baez2018props}: string diagrams naturally keep track of structural properties, such as interconnectivity; they factor out irrelevant topological information that standard algebraic syntax needs to keep explicit; variable-sharing---the relational form of composition for systems---is depicted simply by wiring different components together.

As a result, a wealth of recent developments in computer science and beyond have adopted relations and their diagrammatic notation as a unifying language to reason about a broad range of systems, from electrical circuits to Petri nets~\cite{baez2018compositional,bonchiGraphicalAffineAlgebra2019,BonchiHPSZ19}. Many of these follow the same methodology. 1) Given a class of systems, find a set of diagrammatic generators from which any system can be specified, using the two available forms of composition. 2) Interpret each of them as a relation between the observable variables of the system that they describe. This defines a structure-preserving mapping---a monoidal functor---from the diagrammatic syntax to the semantics, from the two-dimensional representation of a system to its behaviour. 3) Finally, identify a convenient set of equations between diagrams, from which any semantic equality between the behaviour of the corresponding systems may be derived.

Graphical linear algebra (GLA) is a paradigmatic example of this approach. It provides a diagrammatic syntax to reason compositionally about different types of linear dynamical systems (including for instance traditional \emph{signal flow graphs}) and prove their behavioural equivalence purely diagrammatically. 
The syntax of GLA is generated by the following primitive components:
\begin{equation*}\label{eq:syntax-gla}
    \Bcomult\mid\Bcounit\mid\Bmult\mid\Bunit\mid\Wcomult\mid\Wcounit\mid\Wmult\mid\Wunit\mid\scalar{x} \quad (x\in\field)
\end{equation*}
As relations, the black nodes force all of their ports to share the same value; the white nodes constrain their left ports and the right ports to sum to the same value (or to zero when there are no left/right ports); the final generator, parameterised by an element of the chosen field $\field$, behaves as an amplifier: its right value is $x$ times the left value. Following point 3) of the methodology sketched above, GLA enjoys a sound and complete equational theory for the specified semantics, called the theory of \emph{Interacting Hopf Algebras} ($\IH$). In summary, string diagrams with $n$ ports on the left and $m$ ports on the right, quotiented by the axioms of IH, are precisely linear relations, i.e., linear subspaces of $\field^n\times\field^m$.
 
GLA was the starting point of different extensions, in several directions, two of which play a prominent role in this paper. First, Graphical Affine Algebra, which adds to the syntax a generator $\One$ for the constant $1$. This allows it to express \emph{affine} relations, i.e. affine subspaces of $\field^n\times\field^m$. A corresponding complete equational theory was presented in~\cite{bonchiGraphicalAffineAlgebra2019}. Then, Graphical Polyhedral Algebra (GPA), which assumes that $\field$ is an ordered field and adds a generator $
\tikzset{x=1em, y=2.1ex}
\begin{tikzpicture}
	\begin{pgfonlayer}{nodelayer}
		\node [style=geq] (0) at (0, 0) {};
		\node [style=none] (2) at (-1.25, 0) {};
		\node [style=none] (3) at (1.25, 0) {};
	\end{pgfonlayer}
	\begin{pgfonlayer}{edgelayer}
		\draw [in=0, out=180] (0) to (2.center);
		\draw (3.center) to (0);
	\end{pgfonlayer}
\end{tikzpicture}
}
\tikzset{x=1em, y=1.5ex}
$ for this order. The resulting graphical calculus can express all polyhedral relations, i.e., polyhedra\footnote{For the case of $\mathbb{R}$, these include the usual polytopes, which are bounded subsets of $\mathbb{R}^n\times\mathbb{R}^m$, as well as proper polyhedra, which may have unbounded faces.} in $\field^n\times\field^m$, and also comes with its own complete axiomatisation.

In this paper, we define the most expressive member of the GLA family tree to date: Graphical Piecewise-Linear Algebra (GPLA) is a hybrid of symbolic and diagrammatic syntax for piecewise-linear (pl) relations---finite unions of polyhedra in $\field^n\times\field^m$---and a corresponding complete equational theory. We argue below that the proposed language strikes a convincing balance between structure and expressiveness. It is a simple extension of GPA~\cite{bonchiDiagrammaticPolyhedralAlgebra2021}, yet for $\field = \mathbb{R}$, it is sufficiently powerful to approximate any submanifold of $\mathbb{R}^n$ arbitrarily closely.

Furthermore, this extension completes a research program initiated in parallel with the birth of GLA~\cite{baez2018compositional,bonchiGraphicalAffineAlgebra2019,boisseauStringDiagrammaticElectrical2021}: its chief purpose was to give the informal graphical notation for electrical circuits a formal, compositional interpretation, with a corresponding equational theory. 

Until now however, the category-theoretic setting could only accommodate components with a linear (more precisely, affine) behaviour, such as resistors, inductors, capacitors, voltage and current sources.
GPLA finally makes it possible to reason equationally about electronic components, such as ideal diodes and transistors.
Even when the idealised physical behaviour of these components is not necessarily piecewise-linear,
GPLA is theoretically expressive enough to approximate it as closely as necessary. 
Indeed, piecewise-linear approximations of transistor behaviour have been proposed to bypass the unavoidable abstraction leaks of purely digital circuits~\cite{stephan1993physically}.
In this context, GPLA can serve as a form of abstract interpretation for electronic circuits, with adjustable precision to allow for the intended semantics to be as physically realistic as desired.
Of course, in practice, working with large diagrams can be prohibitive. But this is a limitation
shared by all members of the Graphical Algebra family, and developing convenient tools and
techniques for diagrammatic reasoning is an active research area.
Our main thrust is that piecewise-linearity provides the appropriate level of structure, where
general relations are too flexible to come with a useful equational theory, and linear relations are
too rigid to accommodate diodes and other electronic components. 

Finally, a remark about syntax. While it is possible to make the language purely diagrammatic, we found that what one gains in purity one loses in complexity. Ultimately, the hybrid syntax of union and diagrams is more convenient to manipulate and intuitive to read. In fact, this is not the first time that sums of diagrams appear in the literature~\cite{cvitanovic2008group}. Nevertheless, one of our central technical contributions is the rigorous definition of a syntax blending diagrams and binary joins, and the corresponding notion of equational theory.

\paragraph{Outline.} In \cref{sec:prelim} we recall the necessary mathematical
background, the fundamentals of diagrammatic syntax, and the language of Graphical Polyhedral Algebra
(GPA).
In \cref{sec:smlt}, we extend the diagrammatic syntax with unions and define the notion of symmetric monoidal semi-lattice theory. From
there, in \cref{sec:pl-relations}, we extend GPA with unions, to capture piecewise-linear relations,
and give this new language a theory that we prove is complete (\cref{thm:completeness}). This our
main technical contribution. In \cref{sec:alternative-syntax}, we explore alternative languages for
piecewise-linear relations, and show that they are all equally expressive. Finally, in \cref{sec:electronic}, we extend the compositional re-interpretation of electrical circuits from~\cite{boisseauStringDiagrammaticElectrical2021} to include electronic components, namely diodes and transistors. 


\section{Preliminaries}\label{sec:prelim}

Informally, our starting point is a simple diagrammatic language of circuits built from the following generators:
\begin{equation}\label{eq:syntax}
    \Bcomult\mid\Bcounit\mid\Bmult\mid\Bunit
    \mid\Wcomult\mid\Wcounit\mid\Wmult\mid\Wunit
    \mid\One
    \mid\scalar{\geq}
    \mid\scalar{x} \quad (x\in\field)
\end{equation}
We will explain how these basic components can be wired together and give them a formal interpretation.

\subsection{Props and Symmetric Monoidal Theories}\label{sec:smt}

The mathematical backbone of our approach is the notion of product and permutations category (prop), a 
structure which generalises standard algebraic theories~\cite{bonchi2018deconstructing}. Formally, a \emph{prop} is a
strict symmetric monoidal category (SMC) whose objects are the natural numbers and where the
monoidal product $\oplus$ on objects is given by addition. Equivalently, it is a strict SMC whose
objects are all monoidal products of a single generating object.
\emph{Prop morphisms} are strict
symmetric monoidal functors that act as the identity on objects.

Following an established methodology, we will define two props: $\Syn$ and $\Sem$, for the syntax and semantics respectively. To guarantee a compositional interpretation, we require $\sem{\cdot}\from\Syn\to\Sem$, the mapping of terms to their intended semantics, to be a prop morphism.

Typically, the syntactic prop $\Syn$ is freely generated from a \emph{monoidal signature} $\Sigma$,
i.e. a set of arrows $g\from m \to n$.
In this case, we use the notation $\freeProp{\Sigma}$ and $\Syn$ interchangeably.
Morphisms of $\freeProp{\Sigma}$ are terms of an $(\mathbb{N},\mathbb{N})$-sorted syntax, whose constants are elements of $\Sigma$ and
whose operations are the usual composition $(-);(-)\from \Syn(n,m)\times \Syn(m,l)\to
\Syn(n,l)$ and the monoidal product $(-)\oplus(-)\from \Syn(n_1,m_1)\times \Syn(n_2,m_2)\to
\Syn(n_1+n_2,m_1+m_2)$, quotiented by the laws of SMCs.
But this quotient is cumbersome and unintuitive to work with.

This is why we will prefer a different representation.
With their two forms of composition, monoidal categories admit a natural two-dimensional graphical notation of \emph{string diagrams}.
The idea is that an arrow $c\from n\to m$ of $\freeProp{\Sigma}$ is better represented
as a box with $n$ ordered wires on the left and $m$ on the left.
We can compose these diagrams in two different ways:
horizontally, by connecting the right wires of one diagram to the left wires of another,
and vertically by juxtaposing two diagrams:
\[c\poi d = \quad 
\tikzset{x=1em, y=2.1ex}
\InputIfFileExists{comp-sequential-n.tikz}{}{\input{./tikz/comp-sequential-n.tikz}}
\tikzset{x=1em, y=1.5ex}
 \qquad \qquad d_1\oplus d_2 = \quad 
\tikzset{x=1em, y=2.1ex}
\InputIfFileExists{comp-parallel-n.tikz}{}{\input{./tikz/comp-parallel-n.tikz}}
\tikzset{x=1em, y=1.5ex}
\]
where the labelled wire $\idone^n$ is syntactic sugar for a stack of $n$ wires.
%
The identity $id_1\from 1 \to
1$ is denoted as a plain wire $\idone$, the symmetry $\sigma_{1,1}\from 2 \to 2$ as a wire
crossing $\sym$, and the unit for $\oplus$, $id_0\from 0 \to 0$, as the empty diagram $\idzero$.
With this representation the laws of SMCs become diagrammatic tautologies.

Once we have defined $\sem{\cdot}\from\Syn\to\Sem$, it is natural to look for equations to reason about semantic equality directly on the diagrams themselves. Given a set of equations $E$, i.e., a set containing pairs of arrows of the same type, we write $\eqE{E}$ for the smallest congruence wrt the two composition operations $;$ and $\oplus$.   We say that $\eqE{E}$ is \emph{sound} if $c\eqE{E} d$ implies $\sem{c} = \sem{d}$. It is moreover \emph{complete} when the converse implication holds. We call a pair $(\Sigma, E)$ a \emph{symmetric monoidal theory} (SMT) and we can form the prop $\freeProp{\Sigma}_{/E}$ obtained by quotienting $\freeProp{\Sigma}$ by $\eqE{E}$. There is then a prop morphism $q\from \freeProp{\Sigma}\to\freeProp{\Sigma}_{/E}$ witnessing this quotient. 

We may also wonder what the expressive power of our diagrammatic language is. In terms of props we look to characterise precisely the image $\mathsf{Im}(\sem{\cdot})$ of the syntax via $\sem{\cdot}$.

The situation for a sound and complete SMT is summarised in the commutative diagram below right.

\begin{minipage}{0.5\textwidth}
\noindent Soundness simply means that $\sem{\cdot}$ factors as $s\circ q$ through $\freeProp{\Sigma}_{/E}$ and completeness means that $s$ is a faithful prop morphism.
\end{minipage}
\begin{minipage}{0.5\textwidth}
\begin{tikzcd}
\freeProp{\Sigma}_{/E} \arrow[rr, "\cong"] \arrow[rrd, "s"]                  &  & \mathsf{Im}(\sem{\cdot}) \arrow[d, "i", hook] \\
\Syn = \freeProp{\Sigma} \arrow[u, "q", two heads] \arrow[rr, "\sem{\cdot}"'] &  & \Sem                                 
\end{tikzcd}
\end{minipage}

Typically, our semantic prop $\Sem$ will be (a subcategory of) the category of sets and relations.
\begin{definition}\label{def:relk}
Let $\field$ be a field. $\RelX{\field}$ is the prop 
\begin{itemize}
\item whose arrows $n\to m$ are relations $R\subseteq \field^n\times \field^m$,
\item with composition given by $R;S = \{(x,z)\mid \exists y.\; (x,y)\in R \land (y,z)\in S\}$, 
for $R\from n \to m$, $S\from m\to l$, and identity $n\to n$ the diagonal $\{(x,x)\mid x\in \field^n\}$,
\item monoidal product given by
\[R_1\oplus R_2 = \left\{\left(\begin{pmatrix}x_1\\ x_2\end{pmatrix}, \begin{pmatrix}y_1\\ y_2\end{pmatrix}\right)\mid (x_1,y_1)\in R_1 \land (x_2,y_2)\in R_2\right\}\]
for $R_1\from n_1 \to m_1$ and $R_2\from n_2\to m_2$,
\item symmetry $n+m \to m+n$, the relation $\left\{\left(\begin{pmatrix}x\\ y\end{pmatrix}, \begin{pmatrix}y\\ x\end{pmatrix}\right)\mid (x,y)\in \field^n\times\field^m\right\}\text{.}$
\end{itemize}
\end{definition}

\subsection{Ordered Props and Symmetric Monoidal Inequality Theories}\label{sec:smit}

Our semantic prop---$\RelX{\field}$---carries additional structure that we wish to lift to the syntax: as subsets of $\field^n\times\field^m$, relations $n\to m$ can be ordered by inclusion. The corresponding structure is that of an \emph{ordered prop}, a prop enriched over the category of posets, whose composition and monoidal product are monotone maps. 

If props can be presented by SMTs, ordered props can be presented by \emph{symmetric monoidal inequality theories} (SMIT). Formally, the data of a SMIT is the same as that of a SMT: a signature $\Sigma$ 
and a set $I$ of pairs $c,d\from n\to m$ of $\freeProp{\Sigma}$-arrows  of the same type, that we now read as \emph{inequalities} $c\leq d$. 

As for plain props, we can construct an ordered prop from a SMIT by building the free prop $\freeProp{\Sigma}$ and passing to a quotient $\freeProp{\Sigma}_{/I}$. First, we build a preorder on each homset by closing $I$ under $\oplus$ and taking the reflexive and transitive closure of the resulting relation. Then, we obtain the free ordered prop $\freeProp{\Sigma}_{/I}$ by quotienting the resulting preorder by imposing anti-symmetry.

SMITs subsume SMTs, since every SMT can be presented as a SMIT, by splitting each equation into two inequalities. 
We will refer to both simply as \emph{theories} and their defining inequalities as \emph{axioms}. When referring to a sound and complete theory, we will also use the term \emph{axiomatisation}, as is standard in the literature.

\subsection{Graphical Polyhedral Algebra}\label{sec:polyhedra}

We now assume that $\field$ is an ordered field, that is, a field equipped with a total order $\geq$ compatible with the field operations in the following sense: for all $x,y,z\in\field$, \emph{i)} if $x\geq y$ then $x+z\geq y+z$, and \emph{ii)} if $x\geq 0$  and $y\geq 0$ then $xy\geq 0$.

Following~\cite{bonchiDiagrammaticPolyhedralAlgebra2021}, from the generators in~\eqref{eq:syntax}, we define a prop, give it a semantics in $\RelX{\field}$, characterise the image of the semantic functor, and describe an axiomatisation for the specified semantics. 
\begin{itemize}
\item For $\Sigma^+_\geq = \{\Bcomult,\Bcounit,\Bmult,\Bunit
    ,\Wcomult,\Wcounit,\Wmult,\Wunit
    ,\One
    ,\scalar{\geq}
    ,\scalar{r} (r\in\field)\}$ define $\sem{\cdot}\from \mathsf{P}\Sigma^+_{\geq}  \to \RelX{\field}$ to be the prop morphism given by
\begin{equation*}\label{eq:mat-semantics}
\def\arraystretch{1.3}
\begin{array}{cclcccl}
    \sem{\Bcomult} &\df& \left\{\left(x,\begin{pmatrix} x\\ x \end{pmatrix}\right) \mid x \in \field \right\}
    &&
    \sem{\Bcounit} &\df& \{(x,\bullet) \mid x \in \field \}
    \\
    \sem{\Bmult} &\df& \left\{\left(\begin{pmatrix} x\\ x \end{pmatrix},x\right) \mid x \in \field \right\}
    &&
    \sem{\Bunit} &\df& \{(\bullet,x) \mid x \in \field \}
    \\
    \sem{\Wcomult} &\df& \left\{\left(x+y,\begin{pmatrix} x\\ {y} \end{pmatrix}\right) \mid x,{y} \in \field \right\}
    &&
    \sem{\Wcounit} &\df& \{(0,\bullet)\}
    \\
    \sem{\Wmult} &\df& \left\{\left(\begin{pmatrix} x\\ {y} \end{pmatrix},{x+y}\right) \mid x,{y} \in \field \right\}
    &&
    \sem{\Wunit} &\df& \{(\bullet, {0})  \}
    \\
    \sem{\scalar{k}} &\df&
        \multicolumn{4}{l}{\{(x, k \cdot x) \mid x \in \field \} \text{\; for } k \in \field}
    \\
    \sem{
\tikzset{x=1em, y=2.1ex}
\begin{tikzpicture}
	\begin{pgfonlayer}{nodelayer}
		\node [style=geq] (0) at (0, 0) {};
		\node [style=none] (2) at (-1.25, 0) {};
		\node [style=none] (3) at (1.25, 0) {};
	\end{pgfonlayer}
	\begin{pgfonlayer}{edgelayer}
		\draw [in=0, out=180] (0) to (2.center);
		\draw (3.center) to (0);
	\end{pgfonlayer}
\end{tikzpicture}
}
\tikzset{x=1em, y=1.5ex}
} &\df& \{(x,y)\in \field\times \field\mid x \geq y\}
    &&
    \sem{\One} &\df& \{(\bullet, 1)\}
\end{array}
\end{equation*}
\item The image of $\mathsf{P}{\Sigma^+_\geq}$ by $\sem{\cdot}$ is the prop whose arrows $n\to m$ are finitely generated \emph{polyhedra} of $\field^n\times\field^m$, i.e., subsets of the form 
\[\left\{(x,y)\in \field^n\times\field^m \mid A\begin{pmatrix}x\\ y\end{pmatrix} + b \geq 0 \right\}\] 
for some matrix $A$ and some vector $b$
(see \cite{bonchiDiagrammaticPolyhedralAlgebra2021} for more details, in particular the appendix for the proof that these form a prop).

\item $\IHpoly$ provides an axiomatisation of polyhedral relations~\cite[Corollary 25]{bonchiDiagrammaticPolyhedralAlgebra2021}; it can be found in the first four blocks of \cref{fig:ih-pl}.
\end{itemize}
\begin{example}[Duality]
Two diagrams play a special role in this paper: the half turns $
\tikzset{x=1em, y=2.1ex}
\InputIfFileExists{cup-black-small.tikz}{}{\input{./tikz/cup-black-small.tikz}}
\tikzset{x=1em, y=1.5ex}
$ and $
\tikzset{x=1em, y=2.1ex}
\InputIfFileExists{cap-black-small.tikz}{}{\input{./tikz/cap-black-small.tikz}}
\tikzset{x=1em, y=1.5ex}
$, called cup and cap, respectively. Using these and $\sym$, we can build cups and caps for any number $n$ of wires: $
\tikzset{x=1em, y=2.1ex}
\InputIfFileExists{cup-black-small-n.tikz}{}{\input{./tikz/cup-black-small-n.tikz}}
\tikzset{x=1em, y=1.5ex}
$ and $
\tikzset{x=1em, y=2.1ex}
\InputIfFileExists{cap-black-small-n.tikz}{}{\input{./tikz/cap-black-small-n.tikz}}
\tikzset{x=1em, y=1.5ex}
$. 

They allow us to associate a dual $d^{op}\from n \to m$ to any diagram $d\from m \to n$ by turning its left ports into right ports and vice-versa:
\begin{equation}\label{eq:flip}
\diagbox{\scriptstyle d^{op}}{n}{m}\; = \;
\tikzset{x=1em, y=2.1ex}
\InputIfFileExists{d-op.tikz}{}{\input{./tikz/d-op.tikz}}
\tikzset{x=1em, y=1.5ex}
 
\end{equation}
Correspondingly, $\sem{d^{op}}$ is the \emph{opposite} relation, i.e.  $\sem{d^{op}} = \left\{(y,x)\mid (x,y)\in \sem{d}\right\}$. We will use of a suggestive mirror notation to denote the dual of a given generator: $\coscalar{r} := \left(\scalar{r}\right)^{op}$, $
\tikzset{x=1em, y=2.1ex}
\begin{tikzpicture}[baseline=-.5ex]
	\begin{pgfonlayer}{nodelayer}
		\node [style=none] (0) at (0.25, 0) {};
		\node [style=none] (1) at (-0.75, 0) {};
		\node [style=none] (2) at (0.25, 0.25) {};
		\node [style=none] (3) at (0.25, -0.25) {};
	\end{pgfonlayer}
	\begin{pgfonlayer}{edgelayer}
		\draw (0.center) to (1.center);
		\draw (3.center) to (2.center);
	\end{pgfonlayer}
\end{tikzpicture}
}
\tikzset{x=1em, y=1.5ex}
 := \left(\One\right)^{op}$ and $
\tikzset{x=1em, y=2.1ex}
\begin{tikzpicture}
	\begin{pgfonlayer}{nodelayer}
		\node [style=leq] (0) at (0, 0) {};
		\node [style=none] (2) at (-1.25, 0) {};
		\node [style=none] (3) at (1.25, 0) {};
	\end{pgfonlayer}
	\begin{pgfonlayer}{edgelayer}
		\draw [in=0, out=180] (0) to (2.center);
		\draw (3.center) to (0);
	\end{pgfonlayer}
\end{tikzpicture}
}
\tikzset{x=1em, y=1.5ex}
:=\left(
\tikzset{x=1em, y=2.1ex}
\begin{tikzpicture}
	\begin{pgfonlayer}{nodelayer}
		\node [style=geq] (0) at (0, 0) {};
		\node [style=none] (2) at (-1.25, 0) {};
		\node [style=none] (3) at (1.25, 0) {};
	\end{pgfonlayer}
	\begin{pgfonlayer}{edgelayer}
		\draw [in=0, out=180] (0) to (2.center);
		\draw (3.center) to (0);
	\end{pgfonlayer}
\end{tikzpicture}
}
\tikzset{x=1em, y=1.5ex}
\right)^{op}$.

\end{example}

\section{Symmetric Monoidal Semi-Lattice Theories}\label{sec:smlt}

There are several routes to describe piecewise-linear subsets of $\field^n$. In this paper we choose to equip our syntax with a primitive operation of join, in order to describe piecewise-linear sets as (finite) unions of polyhedra.
In the same way that we moved from simple props to ordered props in \cref{sec:smit},
we now move to the setting of semi-lattice-enriched props.

A \emph{$\cup$-prop} is a prop enriched over the category of semi-lattices and join-preserving maps, i.e. whose homsets are semi-lattices, with composition and monoidal product join-preserving maps. The paradigmatic example is $\RelX{\field}$ which is a $\cup$-prop with the union of relations as join.

As we would like to incorporate binary joins into our syntax, we need a new description of the free $\cup$-prop $\freeUProp{\Sigma}$ over a given signature $\Sigma$. 
\begin{itemize}
\item The arrows $n\to m$ of $\freeUProp{\Sigma}$ are nonempty finite sets of arrows $n \to m$ of $\freeUProp{\Sigma}$. We use capital letters $C, D\dots$ to range over them. We will also abuse notation slightly, using $c,d\dots$ to refer to singletons $\{c\}, \{d\}\dots$ and writing $d_1\cup \dots \cup d_k$ for the set $\{d_1, \dots, d_k\}$.
\item The composition of $C\from n\to m$ and $D\from m\to l$ is given by $C\poi D = \{c\poi d \mid c\in C, d\in D\}$ where $c\poi d$ denotes composition in $\mathsf{P}_\Sigma$. The identity over $n$ is the singleton $\{id_n\}$.
\item The monoidal product of $D_1\from n_1\to m_1$ and $D_2\from n_2\to m_2$ is given by $D_1\oplus D_2 = \{d_1\oplus d_2\mid d_1\in D_1, d_2\in D_2 \}$ where $d_1\oplus d_2$ is the monoidal product in $\freeUProp{\Sigma}$.
\item For the enrichment, each homset $\freeUProp{\Sigma}(n,m)$ is a semi-lattice with union as join. By definition, composition and monoidal product distribute over union and define join-preserving maps $(-)\poi (-)\from \freeUProp{\Sigma}(n,m)\times \freeUProp{\Sigma}(m,l)\to \freeUProp{\Sigma}(n,l)$ and $(-)\oplus(-)\from \freeUProp{\Sigma}(n_1,m_1)\times \freeUProp{\Sigma}(n_2,m_2)\to \freeUProp{\Sigma}(n_1+n_2,m_1+m_2)$
\end{itemize}

We now define a corresponding notion of theory for $\cup$-props. A \emph{symmetric monoidal (semi-)lattice theory} (SMLT) is the data of a signature $\Sigma$ and a set $E$ of equations: formally the latter is a set of pairs $(C, D)$ of arrows $C,D\from n \to m$ from $\freeUProp{\Sigma}$. We will write the elements of $E$ as equations of the form $ \bigcup_{c\in C}c = \bigcup_{d\in D} d$. We now explain how to define a $\cup$-prop $\freeUProp{\Sigma}_{/E}$ from the data of an SMLT $(\Sigma, E)$. As for SMTs, we can build the smallest congruence $\eqE{E}$ wrt to $\poi$ and $\oplus$, which equates the pairs in $E$. Then define $\freeUProp{\Sigma}_{/E}$ to be the quotient of $\freeUProp{\Sigma}$ by $\eqE{E}$. That this is a well-defined $\cup$-prop follows again from the distributivity of the composition and monoidal product over unions.

Note that the semi-lattice structure allows us to define an order over the homsets of any $\cup$-prop, making it into an ordered prop: we write $C\subseteq D$ as a shorthand for $C\cup D = D$. We will also use $C\leqE{E} D$ for $C\cup D \eqE{E} D$ in $\freeUProp{\Sigma}_{/E}$. (We prefer this notation to avoid the confusion with the order $\geq$ on $\field$ itself.)

\begin{remark}[Reasoning in $\cup$-props]
The reader familiar with string diagrams and equational reasoning might be surprised by certain features of derivations that combine diagrammatic and traditional syntax (joins, in this case). We would like to clarify one particular point. When we want to use an equality of the form $d = d_1 \cup d_2$ inside a term of the form $c_1 \cup c_2 \cup c$, we need to identify a context $C[-]$ common to $c_1$ and $c_2$ such that $c_1 = C[d_1]$ and $c_2 = C[d_2]$. Then we are allowed to conclude that $c_1 \cup c_2 \cup c = C[d] \cup c$. An example of this form of reasoning can be found in the proof of \cref{lem:pl-nf}, which we reproduce here: we apply the equality $\Bunit \;\myeq{total} \quad 
\tikzset{x=1em, y=2.1ex}
\begin{tikzpicture}
	\begin{pgfonlayer}{nodelayer}
		\node [style=leq] (0) at (-4, 0) {};
		\node [style=white] (2) at (-5.25, 0) {};
		\node [style=none] (3) at (-2.75, 0) {};
	\end{pgfonlayer}
	\begin{pgfonlayer}{edgelayer}
		\draw [in=0, out=180] (0) to (2);
		\draw (3.center) to (0);
	\end{pgfonlayer}
\end{tikzpicture}
}
\tikzset{x=1em, y=1.5ex}
 \cup \; 
\tikzset{x=1em, y=2.1ex}
\begin{tikzpicture}
	\begin{pgfonlayer}{nodelayer}
		\node [style=geq] (0) at (-4, 0) {};
		\node [style=white] (2) at (-5.25, 0) {};
		\node [style=none] (3) at (-2.75, 0) {};
	\end{pgfonlayer}
	\begin{pgfonlayer}{edgelayer}
		\draw [in=0, out=180] (0) to (2);
		\draw (3.center) to (0);
	\end{pgfonlayer}
\end{tikzpicture}
}
\tikzset{x=1em, y=1.5ex}
$ in
\begin{equation*}

\tikzset{x=1em, y=2.1ex}
\InputIfFileExists{Y+.tikz}{}{\input{./tikz/Y+.tikz}}
\tikzset{x=1em, y=1.5ex}
 \cup  \;
\tikzset{x=1em, y=2.1ex}
\InputIfFileExists{Y-.tikz}{}{\input{./tikz/Y-.tikz}}
\tikzset{x=1em, y=1.5ex}
 = 
\tikzset{x=1em, y=2.1ex}
\InputIfFileExists{Y+UY-.tikz}{}{\input{./tikz/Y+UY-.tikz}}
\tikzset{x=1em, y=1.5ex}
\;\myeq{total} \; 
\tikzset{x=1em, y=2.1ex}
\InputIfFileExists{Y+UY-tot.tikz}{}{\input{./tikz/Y+UY-tot.tikz}}
\tikzset{x=1em, y=1.5ex}

\end{equation*}
Note that, to clarify the common context to the reader, we will often use the intermediate notation $C[d_1\cup d_2]$, as we did in the first step above.
\end{remark}

\section{The Theory of Piecewise-Linear Relations}
\label{sec:pl-relations}

\subsection{Syntax and Semantics}

For piecewise-linear relations we retain the same signature $\Sigma_\geq^+$ and consider $\freeUProp(\Sigma_\geq^+)$, the free $\cup$-prop over it. As we saw, its morphisms are nonempty finite sets of diagrams of $\freeProp\Sigma_\geq^+$. This is our syntax. 

On the semantic side, we now need to extend the functor $\sem{\cdot}$ to have $\freeProp\Sigma_\geq^+$ as domain, retaining $\RelX{\field}$ as codomain. Concretely, since we already know how to assign a relation to each diagram of  $\freeProp\Sigma_\geq^+$, we only need to specify how to interpret finite sets of such diagrams: unsurprisingly, we set
\[\sem{\{d_1, \dots, d_n\}} := \sem{d_1}\cup \dots \cup \sem{d_n}\]
This is join-preserving by construction, and remains monoidal and functorial.

By definition, we call piecewise-linear (pl) any relation in the image of this functor, i.e., any relation that is a finite union of polyhedral relations. As far as we know, this is the first time that this notion appears in print. However, it does capture our intuitive notion of piecewise-linearity as submanifolds of $\field^n$ that can be subdivided into linear subspaces.

\subsection{Equational Theory}\label{sec:axioms}

$\IHpl$, the SMLT of pl relations, is presented in \cref{fig:ih-pl}. The first block is the theory of matrices/linear maps; the second block, $\IH$, axiomatises all linear relations; the third block axiomatises the behaviour of the order $
\tikzset{x=1em, y=2.1ex}
\begin{tikzpicture}
	\begin{pgfonlayer}{nodelayer}
		\node [style=geq] (0) at (0, 0) {};
		\node [style=none] (2) at (-1.25, 0) {};
		\node [style=none] (3) at (1.25, 0) {};
	\end{pgfonlayer}
	\begin{pgfonlayer}{edgelayer}
		\draw [in=0, out=180] (0) to (2.center);
		\draw (3.center) to (0);
	\end{pgfonlayer}
\end{tikzpicture}
}
\tikzset{x=1em, y=1.5ex}
$; the fourth, deals with the affine fragment of the theory, axiomatising the behaviour of the constant $
\tikzset{x=1em, y=2.1ex}
\begin{tikzpicture}[baseline=-.5ex]
	\begin{pgfonlayer}{nodelayer}
		\node [style=none] (0) at (-0.75, -0) {};
		\node [style=none] (1) at (0.25, -0) {};
		\node [style=none] (2) at (-0.75, 0.25) {};
		\node [style=none] (3) at (-0.75, -0.25) {};
	\end{pgfonlayer}
	\begin{pgfonlayer}{edgelayer}
		\draw (0.center) to (1.center);
		\draw (3.center) to (2.center);
	\end{pgfonlayer}
\end{tikzpicture}}
\tikzset{x=1em, y=1.5ex}
$. Taken together, those four blocks constitute $\IHpoly$, an axiomatisation of polyhedral relations---we refer the reader to~\cite{bonchiDiagrammaticPolyhedralAlgebra2021} for more details on this fragment.

The key addition of $\IHpl$ is the last block, containing the axiom of \emph{totality}, which states that any real number belongs to the non-negative \emph{or} to the non-positive fragment of $\field$. Remarkably, this simple axiom is the only one we need to add to $\IHpoly$ to obtain a complete theory for pl relations. Its soundness is simply a consequence of the definition of an ordered field: the order is assumed to be total in the sense that, for any $x,y\in\field$ we have $x\leq y$ or $y\leq x$. Take $y=0$ to recover the last axiom of $\IHpl$.

\begin{figure}
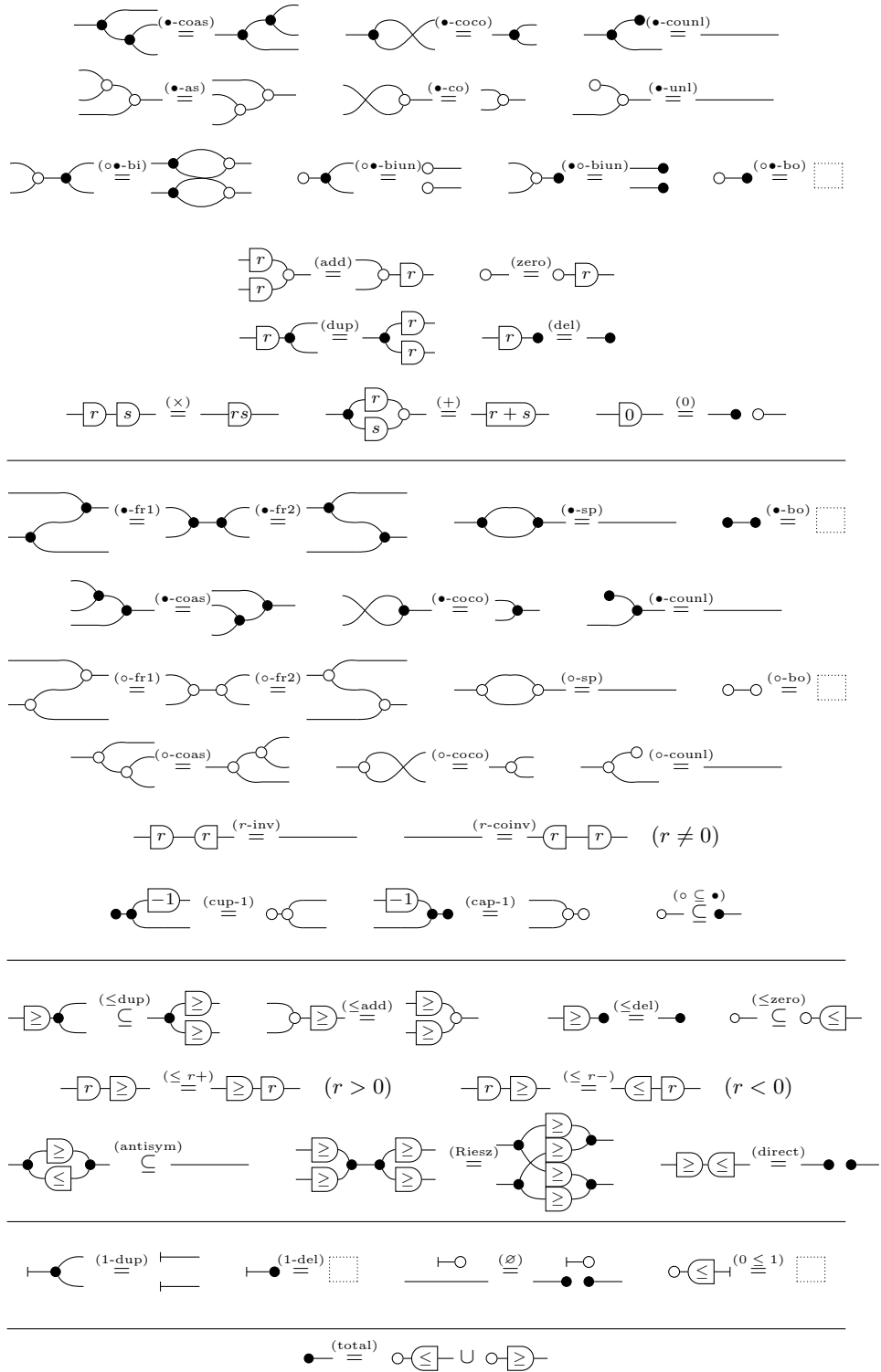

\begin{equation*}

\tikzset{x=1em, y=2.1ex}
\InputIfFileExists{ax/copy-associative.tikz}{}{\input{./tikz/ax/copy-associative.tikz}}
\tikzset{x=1em, y=1.5ex}
\;\;\myeq{$\bullet$-coas}\;\; 
\tikzset{x=1em, y=2.1ex}
\InputIfFileExists{ax/copy-associative-1.tikz}{}{\input{./tikz/ax/copy-associative-1.tikz}}
\tikzset{x=1em, y=1.5ex}
\qquad  
\tikzset{x=1em, y=2.1ex}
\InputIfFileExists{ax/copy-commutative.tikz}{}{\input{./tikz/ax/copy-commutative.tikz}}
\tikzset{x=1em, y=1.5ex}
\;\;\myeq{$\bullet$-coco}\;\;
\tikzset{x=1em, y=2.1ex}
\InputIfFileExists{generators/copy.tikz}{}{\input{./tikz/generators/copy.tikz}}
\tikzset{x=1em, y=1.5ex}
\qquad  
\tikzset{x=1em, y=2.1ex}
\InputIfFileExists{ax/copy-unital-left.tikz}{}{\input{./tikz/ax/copy-unital-left.tikz}}
\tikzset{x=1em, y=1.5ex}
\myeq{$\bullet$-counl}\;
\tikzset{x=1em, y=2.1ex}
\InputIfFileExists{generators/id.tikz}{}{\input{./tikz/generators/id.tikz}}
\tikzset{x=1em, y=1.5ex}
\end{equation*}
\begin{equation*}

\tikzset{x=1em, y=2.1ex}
\InputIfFileExists{ax/add-associative.tikz}{}{\input{./tikz/ax/add-associative.tikz}}
\tikzset{x=1em, y=1.5ex}
\;\myeq{$\bullet$-as}\;\: 
\tikzset{x=1em, y=2.1ex}
\InputIfFileExists{ax/add-associative-1.tikz}{}{\input{./tikz/ax/add-associative-1.tikz}}
\tikzset{x=1em, y=1.5ex}
\qquad  
\tikzset{x=1em, y=2.1ex}
\InputIfFileExists{ax/add-commutative.tikz}{}{\input{./tikz/ax/add-commutative.tikz}}
\tikzset{x=1em, y=1.5ex}
\myeq{$\bullet$-co}\;\;\,
\tikzset{x=1em, y=2.1ex}
\InputIfFileExists{generators/add.tikz}{}{\input{./tikz/generators/add.tikz}}
\tikzset{x=1em, y=1.5ex}
\qquad  
\tikzset{x=1em, y=2.1ex}
\InputIfFileExists{ax/add-unital-left.tikz}{}{\input{./tikz/ax/add-unital-left.tikz}}
\tikzset{x=1em, y=1.5ex}
\;\myeq{$\bullet$-unl}\;
\tikzset{x=1em, y=2.1ex}
\InputIfFileExists{generators/id.tikz}{}{\input{./tikz/generators/id.tikz}}
\tikzset{x=1em, y=1.5ex}

\end{equation*}
\vspace{1pt}
\begin{equation*}

\tikzset{x=1em, y=2.1ex}
\InputIfFileExists{ax/add-copy-bimonoid.tikz}{}{\input{./tikz/ax/add-copy-bimonoid.tikz}}
\tikzset{x=1em, y=1.5ex}
\;\;\myeq{$\circ\bullet$-bi}\;\;
\tikzset{x=1em, y=2.1ex}
\InputIfFileExists{ax/add-copy-bimonoid-1.tikz}{}{\input{./tikz/ax/add-copy-bimonoid-1.tikz}}
\tikzset{x=1em, y=1.5ex}
 \qquad 
\tikzset{x=1em, y=2.1ex}
\InputIfFileExists{ax/add-copy-bimonoid-unit.tikz}{}{\input{./tikz/ax/add-copy-bimonoid-unit.tikz}}
\tikzset{x=1em, y=1.5ex}
 \;\;\;\myeq{$\circ\bullet$-biun}\;\;\; 
\tikzset{x=1em, y=2.1ex}
\InputIfFileExists{ax/add-bimonoid-unit-1.tikz}{}{\input{./tikz/ax/add-bimonoid-unit-1.tikz}}
\tikzset{x=1em, y=1.5ex}
 \qquad 
\tikzset{x=1em, y=2.1ex}
\InputIfFileExists{ax/add-copy-bimonoid-counit.tikz}{}{\input{./tikz/ax/add-copy-bimonoid-counit.tikz}}
\tikzset{x=1em, y=1.5ex}
 \;\;\,\myeq{$\bullet\circ$-biun}\;\;\; 
\tikzset{x=1em, y=2.1ex}
\InputIfFileExists{ax/add-copy-bimonoid-counit-1.tikz}{}{\input{./tikz/ax/add-copy-bimonoid-counit-1.tikz}}
\tikzset{x=1em, y=1.5ex}
\qquad
\tikzset{x=1em, y=2.1ex}
\InputIfFileExists{ax/bone-white-black.tikz}{}{\input{./tikz/ax/bone-white-black.tikz}}
\tikzset{x=1em, y=1.5ex}
\;\;\myeq{$\circ\bullet$-bo}\;\;\;
\tikzset{x=1em, y=2.1ex}
\InputIfFileExists{empty-diag.tikz}{}{\input{./tikz/empty-diag.tikz}}
\tikzset{x=1em, y=1.5ex}

\end{equation*}
\vspace{1pt}
\begin{align*}
  
\tikzset{x=1em, y=2.1ex}
\InputIfFileExists{ax/reals-add.tikz}{}{\input{./tikz/ax/reals-add.tikz}}
\tikzset{x=1em, y=1.5ex}
\;\myeq{add}\;
\tikzset{x=1em, y=2.1ex}
\InputIfFileExists{ax/reals-add-1.tikz}{}{\input{./tikz/ax/reals-add-1.tikz}}
\tikzset{x=1em, y=1.5ex}
 \qquad 
\tikzset{x=1em, y=2.1ex}
\InputIfFileExists{generators/zero.tikz}{}{\input{./tikz/generators/zero.tikz}}
\tikzset{x=1em, y=1.5ex}
\;\myeq{zero}\;
\tikzset{x=1em, y=2.1ex}
\InputIfFileExists{ax/reals-zero.tikz}{}{\input{./tikz/ax/reals-zero.tikz}}
\tikzset{x=1em, y=1.5ex}
\\ 
  
\tikzset{x=1em, y=2.1ex}
\InputIfFileExists{ax/reals-copy.tikz}{}{\input{./tikz/ax/reals-copy.tikz}}
\tikzset{x=1em, y=1.5ex}
\;\myeq{dup}\; 
\tikzset{x=1em, y=2.1ex}
\InputIfFileExists{ax/reals-copy-1.tikz}{}{\input{./tikz/ax/reals-copy-1.tikz}}
\tikzset{x=1em, y=1.5ex}
 \qquad 
\tikzset{x=1em, y=2.1ex}
\InputIfFileExists{ax/reals-delete.tikz}{}{\input{./tikz/ax/reals-delete.tikz}}
\tikzset{x=1em, y=1.5ex}
\;\myeq{del}\;
\tikzset{x=1em, y=2.1ex}
\InputIfFileExists{generators/delete.tikz}{}{\input{./tikz/generators/delete.tikz}}
\tikzset{x=1em, y=1.5ex}

\end{align*}
\begin{equation*}
  
\tikzset{x=1em, y=2.1ex}
\InputIfFileExists{ax/reals-multiplication.tikz}{}{\input{./tikz/ax/reals-multiplication.tikz}}
\tikzset{x=1em, y=1.5ex}
\;\myeq{$\times$}\;
\tikzset{x=1em, y=2.1ex}
\InputIfFileExists{ax/reals-multiplication-1.tikz}{}{\input{./tikz/ax/reals-multiplication-1.tikz}}
\tikzset{x=1em, y=1.5ex}
 \qquad    
\tikzset{x=1em, y=2.1ex}
\InputIfFileExists{ax/reals-sum.tikz}{}{\input{./tikz/ax/reals-sum.tikz}}
\tikzset{x=1em, y=1.5ex}
\;\myeq{$+$}\;
\tikzset{x=1em, y=2.1ex}
\InputIfFileExists{ax/reals-sum-1.tikz}{}{\input{./tikz/ax/reals-sum-1.tikz}}
\tikzset{x=1em, y=1.5ex}
\qquad   
\tikzset{x=1em, y=2.1ex}
\InputIfFileExists{ax/reals-scalar-zero.tikz}{}{\input{./tikz/ax/reals-scalar-zero.tikz}}
\tikzset{x=1em, y=1.5ex}
\;\myeq{$0$}\;
\tikzset{x=1em, y=2.1ex}
\InputIfFileExists{ax/reals-scalar-zero-1.tikz}{}{\input{./tikz/ax/reals-scalar-zero-1.tikz}}
\tikzset{x=1em, y=1.5ex}

\end{equation*}
\vspace{2pt}
\hrule
\vspace{3pt}
\begin{equation*}

\tikzset{x=1em, y=2.1ex}
\InputIfFileExists{ax/copy-Frobenius-left.tikz}{}{\input{./tikz/ax/copy-Frobenius-left.tikz}}
\tikzset{x=1em, y=1.5ex}
\;\;\myeq{$\bullet$-fr1}\;\; 
\tikzset{x=1em, y=2.1ex}
\InputIfFileExists{ax/copy-Frobenius.tikz}{}{\input{./tikz/ax/copy-Frobenius.tikz}}
\tikzset{x=1em, y=1.5ex}
\;\;\myeq{$\bullet$-fr2}\;\; 
\tikzset{x=1em, y=2.1ex}
\InputIfFileExists{ax/copy-Frobenius-right.tikz}{}{\input{./tikz/ax/copy-Frobenius-right.tikz}}
\tikzset{x=1em, y=1.5ex}
 \qquad 
\tikzset{x=1em, y=2.1ex}
\InputIfFileExists{ax/copy-special.tikz}{}{\input{./tikz/ax/copy-special.tikz}}
\tikzset{x=1em, y=1.5ex}
\myeq{$\bullet$-sp}
\tikzset{x=1em, y=2.1ex}
\InputIfFileExists{generators/id.tikz}{}{\input{./tikz/generators/id.tikz}}
\tikzset{x=1em, y=1.5ex}
\qquad 
\tikzset{x=1em, y=2.1ex}
\InputIfFileExists{ax/bone-black.tikz}{}{\input{./tikz/ax/bone-black.tikz}}
\tikzset{x=1em, y=1.5ex}
\;\;\myeq{$\bullet$-bo}\;\;
\tikzset{x=1em, y=2.1ex}
\InputIfFileExists{empty-diag.tikz}{}{\input{./tikz/empty-diag.tikz}}
\tikzset{x=1em, y=1.5ex}

\end{equation*}
\vspace{1pt}
\begin{equation*}

\tikzset{x=1em, y=2.1ex}
\InputIfFileExists{ax/co-copy-associative.tikz}{}{\input{./tikz/ax/co-copy-associative.tikz}}
\tikzset{x=1em, y=1.5ex}
\;\;\myeq{$\bullet$-coas}\;\; 
\tikzset{x=1em, y=2.1ex}
\InputIfFileExists{ax/co-copy-associative-1.tikz}{}{\input{./tikz/ax/co-copy-associative-1.tikz}}
\tikzset{x=1em, y=1.5ex}
\qquad  
\tikzset{x=1em, y=2.1ex}
\InputIfFileExists{ax/co-copy-commutative.tikz}{}{\input{./tikz/ax/co-copy-commutative.tikz}}
\tikzset{x=1em, y=1.5ex}
\;\;\myeq{$\bullet$-coco}\;\;\;
\tikzset{x=1em, y=2.1ex}
\InputIfFileExists{generators/co-copy.tikz}{}{\input{./tikz/generators/co-copy.tikz}}
\tikzset{x=1em, y=1.5ex}
\qquad  
\tikzset{x=1em, y=2.1ex}
\InputIfFileExists{ax/co-copy-unital-left.tikz}{}{\input{./tikz/ax/co-copy-unital-left.tikz}}
\tikzset{x=1em, y=1.5ex}
\myeq{$\bullet$-counl}\;
\tikzset{x=1em, y=2.1ex}
\InputIfFileExists{generators/id.tikz}{}{\input{./tikz/generators/id.tikz}}
\tikzset{x=1em, y=1.5ex}
\end{equation*}
\vspace{2pt}
\begin{equation*}

\tikzset{x=1em, y=2.1ex}
\InputIfFileExists{ax/add-Frobenius-left.tikz}{}{\input{./tikz/ax/add-Frobenius-left.tikz}}
\tikzset{x=1em, y=1.5ex}
\;\;\myeq{$\circ$-fr1}\;\; 
\tikzset{x=1em, y=2.1ex}
\InputIfFileExists{ax/add-Frobenius.tikz}{}{\input{./tikz/ax/add-Frobenius.tikz}}
\tikzset{x=1em, y=1.5ex}
\;\;\myeq{$\circ$-fr2}\;\; 
\tikzset{x=1em, y=2.1ex}
\InputIfFileExists{ax/add-Frobenius-right.tikz}{}{\input{./tikz/ax/add-Frobenius-right.tikz}}
\tikzset{x=1em, y=1.5ex}
 
\qquad 
\tikzset{x=1em, y=2.1ex}
\InputIfFileExists{ax/add-special.tikz}{}{\input{./tikz/ax/add-special.tikz}}
\tikzset{x=1em, y=1.5ex}
\myeq{$\circ$-sp}
\tikzset{x=1em, y=2.1ex}
\InputIfFileExists{generators/id.tikz}{}{\input{./tikz/generators/id.tikz}}
\tikzset{x=1em, y=1.5ex}
\qquad
\tikzset{x=1em, y=2.1ex}
\InputIfFileExists{ax/bone-white.tikz}{}{\input{./tikz/ax/bone-white.tikz}}
\tikzset{x=1em, y=1.5ex}
\;\;\myeq{$\circ$-bo}\;\;
\tikzset{x=1em, y=2.1ex}
\InputIfFileExists{empty-diag.tikz}{}{\input{./tikz/empty-diag.tikz}}
\tikzset{x=1em, y=1.5ex}

\end{equation*}
\vspace{1pt}
\begin{equation*}

\tikzset{x=1em, y=2.1ex}
\InputIfFileExists{ax/co-add-associative-1.tikz}{}{\input{./tikz/ax/co-add-associative-1.tikz}}
\tikzset{x=1em, y=1.5ex}
\;\;\myeq{$\circ$-coas}\; 
\tikzset{x=1em, y=2.1ex}
\InputIfFileExists{ax/co-add-associative.tikz}{}{\input{./tikz/ax/co-add-associative.tikz}}
\tikzset{x=1em, y=1.5ex}
\qquad  
\tikzset{x=1em, y=2.1ex}
\InputIfFileExists{ax/co-add-commutative.tikz}{}{\input{./tikz/ax/co-add-commutative.tikz}}
\tikzset{x=1em, y=1.5ex}
\;\;\;\myeq{$\circ$-coco}\;\; 
\tikzset{x=1em, y=2.1ex}
\InputIfFileExists{generators/co-add.tikz}{}{\input{./tikz/generators/co-add.tikz}}
\tikzset{x=1em, y=1.5ex}
\qquad  
\tikzset{x=1em, y=2.1ex}
\InputIfFileExists{ax/co-add-unital-left.tikz}{}{\input{./tikz/ax/co-add-unital-left.tikz}}
\tikzset{x=1em, y=1.5ex}
\;\myeq{$\circ$-counl}\;
\tikzset{x=1em, y=2.1ex}
\InputIfFileExists{generators/id.tikz}{}{\input{./tikz/generators/id.tikz}}
\tikzset{x=1em, y=1.5ex}
\end{equation*}
\vspace{3pt}
\begin{equation*}
    
\tikzset{x=1em, y=2.1ex}
\InputIfFileExists{ax/scalar-division.tikz}{}{\input{./tikz/ax/scalar-division.tikz}}
\tikzset{x=1em, y=1.5ex}
\;\myeq{$r$-inv}\; 
\tikzset{x=1em, y=2.1ex}
\InputIfFileExists{generators/id.tikz}{}{\input{./tikz/generators/id.tikz}}
\tikzset{x=1em, y=1.5ex}
 \qquad 
\tikzset{x=1em, y=2.1ex}
\InputIfFileExists{generators/id.tikz}{}{\input{./tikz/generators/id.tikz}}
\tikzset{x=1em, y=1.5ex}
\;\myeq{$r$-coinv}\;
\tikzset{x=1em, y=2.1ex}
\InputIfFileExists{ax/scalar-co-division.tikz}{}{\input{./tikz/ax/scalar-co-division.tikz}}
\tikzset{x=1em, y=1.5ex}
\quad (r\neq 0)
\end{equation*}
\vspace{3pt}
\begin{equation*}
   
\tikzset{x=1em, y=2.1ex}
\InputIfFileExists{ax/cup-black-minus-1.tikz}{}{\input{./tikz/ax/cup-black-minus-1.tikz}}
\tikzset{x=1em, y=1.5ex}
\quad\myeq{cup-1}\quad 
\tikzset{x=1em, y=2.1ex}
\InputIfFileExists{ax/cup-white.tikz}{}{\input{./tikz/ax/cup-white.tikz}}
\tikzset{x=1em, y=1.5ex}
 \qquad 
\tikzset{x=1em, y=2.1ex}
\InputIfFileExists{ax/cap-black-minus-1.tikz}{}{\input{./tikz/ax/cap-black-minus-1.tikz}}
\tikzset{x=1em, y=1.5ex}
\quad\myeq{cap-1}\quad 
\tikzset{x=1em, y=2.1ex}
\InputIfFileExists{ax/cap-white.tikz}{}{\input{./tikz/ax/cap-white.tikz}}
\tikzset{x=1em, y=1.5ex}
  \qquad\quad  \Wunit \myleq{$\circ\subseteq\bullet$} \Bunit
\end{equation*}
\vspace{2pt}
\hrule
\vspace{3pt}
\begin{equation*}
   
\tikzset{x=1em, y=2.1ex}
\InputIfFileExists{ax/leq-copy.tikz}{}{\input{./tikz/ax/leq-copy.tikz}}
\tikzset{x=1em, y=1.5ex}
\quad\myleq{$\leq$dup}\; 
\tikzset{x=1em, y=2.1ex}
\InputIfFileExists{ax/copy-leq.tikz}{}{\input{./tikz/ax/copy-leq.tikz}}
\tikzset{x=1em, y=1.5ex}
 \qquad
   
\tikzset{x=1em, y=2.1ex}
\InputIfFileExists{ax/add-leq.tikz}{}{\input{./tikz/ax/add-leq.tikz}}
\tikzset{x=1em, y=1.5ex}
\;\myeq{$\leq$add}\quad 
\tikzset{x=1em, y=2.1ex}
\InputIfFileExists{ax/leq-add.tikz}{}{\input{./tikz/ax/leq-add.tikz}}
\tikzset{x=1em, y=1.5ex}
  \qquad\quad
   
\tikzset{x=1em, y=2.1ex}
\InputIfFileExists{ax/leq-del.tikz}{}{\input{./tikz/ax/leq-del.tikz}}
\tikzset{x=1em, y=1.5ex}
 \;\; \myeq{$\leq$del} \;\Bcounit \qquad
    \Wunit \; \myleq{$\leq$zero} \;
\tikzset{x=1em, y=2.1ex}
\InputIfFileExists{ax/zero-leq.tikz}{}{\input{./tikz/ax/zero-leq.tikz}}
\tikzset{x=1em, y=1.5ex}

\end{equation*}
\begin{equation*}
   
\tikzset{x=1em, y=2.1ex}
\InputIfFileExists{ax/scalar-leq.tikz}{}{\input{./tikz/ax/scalar-leq.tikz}}
\tikzset{x=1em, y=1.5ex}
\quad\myeq{$\leq r+$}\; 
\tikzset{x=1em, y=2.1ex}
\InputIfFileExists{ax/leq-scalar.tikz}{}{\input{./tikz/ax/leq-scalar.tikz}}
\tikzset{x=1em, y=1.5ex}
\quad (r > 0) \qquad\quad  
\tikzset{x=1em, y=2.1ex}
\InputIfFileExists{ax/scalar-leq.tikz}{}{\input{./tikz/ax/scalar-leq.tikz}}
\tikzset{x=1em, y=1.5ex}
\quad\myeq{$\leq r-$}\; 
\tikzset{x=1em, y=2.1ex}
\InputIfFileExists{ax/geq-scalar.tikz}{}{\input{./tikz/ax/geq-scalar.tikz}}
\tikzset{x=1em, y=1.5ex}
\quad (r < 0)
\end{equation*}
\begin{equation*}
   
\tikzset{x=1em, y=2.1ex}
\InputIfFileExists{ax/leq-meet-geq.tikz}{}{\input{./tikz/ax/leq-meet-geq.tikz}}
\tikzset{x=1em, y=1.5ex}
\quad\myleq{antisym}\; 
\tikzset{x=1em, y=2.1ex}
\InputIfFileExists{generators/id.tikz}{}{\input{./tikz/generators/id.tikz}}
\tikzset{x=1em, y=1.5ex}
 \qquad
   
\tikzset{x=1em, y=2.1ex}
\InputIfFileExists{ax/leq-bimonoid.tikz}{}{\input{./tikz/ax/leq-bimonoid.tikz}}
\tikzset{x=1em, y=1.5ex}
\quad\myeq{Riesz}\; 
\tikzset{x=1em, y=2.1ex}
\InputIfFileExists{ax/leq-bimonoid-1.tikz}{}{\input{./tikz/ax/leq-bimonoid-1.tikz}}
\tikzset{x=1em, y=1.5ex}
\qquad
   
\tikzset{x=1em, y=2.1ex}
\InputIfFileExists{ax/leq-geq.tikz}{}{\input{./tikz/ax/leq-geq.tikz}}
\tikzset{x=1em, y=1.5ex}
\;\;\myeq{direct}\; 
\tikzset{x=1em, y=2.1ex}
\InputIfFileExists{ax/bcounit-bunit.tikz}{}{\input{./tikz/ax/bcounit-bunit.tikz}}
\tikzset{x=1em, y=1.5ex}

\end{equation*}
\hrule
\vspace{3pt}
\begin{equation*}

\tikzset{x=1em, y=2.1ex}
\InputIfFileExists{ax/one-copy.tikz}{}{\input{./tikz/ax/one-copy.tikz}}
\tikzset{x=1em, y=1.5ex}
\quad\myeq{1-dup}\quad 
\tikzset{x=1em, y=2.1ex}
\InputIfFileExists{ax/one-twice.tikz}{}{\input{./tikz/ax/one-twice.tikz}}
\tikzset{x=1em, y=1.5ex}
\qquad 
\tikzset{x=1em, y=2.1ex}
\InputIfFileExists{ax/one-delete.tikz}{}{\input{./tikz/ax/one-delete.tikz}}
\tikzset{x=1em, y=1.5ex}
\;\myeq{1-del}\;\;
\tikzset{x=1em, y=2.1ex}
\InputIfFileExists{empty-diag.tikz}{}{\input{./tikz/empty-diag.tikz}}
\tikzset{x=1em, y=1.5ex}
\qquad 
\tikzset{x=1em, y=2.1ex}
\InputIfFileExists{ax/one-false.tikz}{}{\input{./tikz/ax/one-false.tikz}}
\tikzset{x=1em, y=1.5ex}
\;\myeq{$\varnothing$}\; 
\tikzset{x=1em, y=2.1ex}
\InputIfFileExists{one-false-disconnect.tikz}{}{\input{./tikz/one-false-disconnect.tikz}}
\tikzset{x=1em, y=1.5ex}
\qquad 
\tikzset{x=1em, y=2.1ex}
\InputIfFileExists{ax/0-leq-1.tikz}{}{\input{./tikz/ax/0-leq-1.tikz}}
\tikzset{x=1em, y=1.5ex}
\;\;\myeq{$0\leq 1$}\quad
\tikzset{x=1em, y=2.1ex}
\InputIfFileExists{empty-diag.tikz}{}{\input{./tikz/empty-diag.tikz}}
\tikzset{x=1em, y=1.5ex}

\end{equation*}
\vspace{2pt}
\hrule
\vspace{3pt}
\begin{equation*}
\Bunit \;\myeq{total} \quad 
\tikzset{x=1em, y=2.1ex}
\InputIfFileExists{ax/0-leq.tikz}{}{\input{./tikz/ax/0-leq.tikz}}
\tikzset{x=1em, y=1.5ex}
 \cup \; 
\tikzset{x=1em, y=2.1ex}
\InputIfFileExists{ax/0-geq.tikz}{}{\input{./tikz/ax/0-geq.tikz}}
\tikzset{x=1em, y=1.5ex}

\end{equation*}
\caption{Axioms of GPLA.}\label{fig:ih-pl}
\end{figure}

\begin{remark}\label{rmk:compact}
As a consequence of the Frobenius laws ($\bullet$-fr) and of (co)unitality ($\bullet$-un)-($\bullet$-coun), the diagrams $
\tikzset{x=1em, y=2.1ex}
\InputIfFileExists{cup-black-small-n.tikz}{}{\input{./tikz/cup-black-small-n.tikz}}
\tikzset{x=1em, y=1.5ex}
$ and $
\tikzset{x=1em, y=2.1ex}
\InputIfFileExists{cap-black-small-n.tikz}{}{\input{./tikz/cap-black-small-n.tikz}}
\tikzset{x=1em, y=1.5ex}
$ satisfy
\begin{equation}\label{eq:snake}

\tikzset{x=1em, y=2.1ex}
\InputIfFileExists{z-compact.tikz}{}{\input{./tikz/z-compact.tikz}}
\tikzset{x=1em, y=1.5ex}
 \,\pleq\, 
\tikzset{x=1em, y=2.1ex}
\begin{tikzpicture}
	\begin{pgfonlayer}{nodelayer}
		\node [style=none] (0) at (1.75, 0) {};
		\node [style=none] (1) at (-1.75, 0) {};
		\node [style=none] (2) at (1.25, 0.5) {\scriptsize $n$};
	\end{pgfonlayer}
	\begin{pgfonlayer}{edgelayer}
		\draw (0.center) to (1.center);
	\end{pgfonlayer}
\end{tikzpicture}
}
\tikzset{x=1em, y=1.5ex}
 \,\pleq\, 
\tikzset{x=1em, y=2.1ex}
\InputIfFileExists{s-compact.tikz}{}{\input{./tikz/s-compact.tikz}}
\tikzset{x=1em, y=1.5ex}

\end{equation}
for any $n$, the defining equations of a compact closed category. Intuitively, these allow us to forget the direction of wires, as long as we preserve the connectivity of the different components of a diagram. 
In addition, compactness implies the following proposition.
\end{remark}

\begin{proposition}\label{thm:duality}
$C\plleq D$ iff $C^{op}\plleq D^{op}$. 
\end{proposition}
Another important property of compact closed category which we will exploit to simplify the completeness proof is stated in the following proposition. It is an immediate consequence of \eqref{eq:snake}.
\begin{proposition}\label{thm:map-state}
Given $C,D\from m\to n$, $C\plleq D$ iff $
\tikzset{x=1em, y=2.1ex}
\InputIfFileExists{c-effect.tikz}{}{\input{./tikz/c-effect.tikz}}
\tikzset{x=1em, y=1.5ex}
\; \plleq\; 
\tikzset{x=1em, y=2.1ex}
\InputIfFileExists{d-effect.tikz}{}{\input{./tikz/d-effect.tikz}}
\tikzset{x=1em, y=1.5ex}
$
\end{proposition}

\subsection{Completeness Theorem}\label{sec:completeness}

As we stated above,
the axioms in \cref{fig:ih-pl} form a complete theory for pl relations.
We will prove that claim in this section.
Without loss of generality, using \cref{thm:map-state},
we restrict to $n\to 0$ diagrams.

We start by defining appropriate normal forms for polyhedral and pl relations,
and then show that every diagram can be reduced to normal form.

\begin{definition}
    We call hyperplane a nonzero $n \to 1$ affine map $H$ which we write $
\tikzset{x=1em, y=2.1ex}
\begin{tikzpicture}
	\begin{pgfonlayer}{nodelayer}
		\node [style=reg] (9) at (1.75, 0) {$H$};
		\node [style=none] (14) at (2.75, 0) {};
		\node [style=none] (15) at (0.75, 0) {};
		\node [style=none] (20) at (1.75, 0) {};
	\end{pgfonlayer}
	\begin{pgfonlayer}{edgelayer}
		\draw [in=-180, out=0] (15.center) to (20.center);
		\draw (9) to (14);
	\end{pgfonlayer}
\end{tikzpicture}
}
\tikzset{x=1em, y=1.5ex}
$.
    A given hyperplane $H$ defines two half-spaces
    $
\tikzset{x=1em, y=2.1ex}
\InputIfFileExists{hyperplane-pos-halfspace.tikz}{}{\input{./tikz/hyperplane-pos-halfspace.tikz}}
\tikzset{x=1em, y=1.5ex}
$ and $
\tikzset{x=1em, y=2.1ex}
\InputIfFileExists{hyperplane-neg-halfspace.tikz}{}{\input{./tikz/hyperplane-neg-halfspace.tikz}}
\tikzset{x=1em, y=1.5ex}
$,
    as well as an affine subspace $
\tikzset{x=1em, y=2.1ex}
\begin{tikzpicture}
	\begin{pgfonlayer}{nodelayer}
		\node [style=reg] (9) at (1.75, 0) {$H$};
		\node [style=white] (14) at (2.75, 0) {};
		\node [style=none] (15) at (0.75, 0) {};
		\node [style=none] (20) at (1.75, 0) {};
	\end{pgfonlayer}
	\begin{pgfonlayer}{edgelayer}
		\draw [in=-180, out=0] (15.center) to (20.center);
		\draw (9) to (14);
	\end{pgfonlayer}
\end{tikzpicture}
}
\tikzset{x=1em, y=1.5ex}
$.
    Since inequality is not strict, the half-spaces include the affine subspace.
\end{definition}

In \cite[Theorem 14]{bonchiDiagrammaticPolyhedralAlgebra2021},
polyhedral relations have a normal form given by a set of inequations
of the form $A_i x + b_i \geq 0$.
In other words, the normal form is given by an intersection of half-spaces.
For our purposes we define a related but slightly different normal form.

\begin{definition}\label{def:poly-nf}
    A $\freeProp\Sigma_\geq^+$-diagram $d\from n \to 0$ is in polyhedral normal form if
    there are hyperplanes $H_i$ and diagrams
    $\circleeffect{d_i}\in\big\{\Wcounit\,, 
\tikzset{x=1em, y=2.1ex}
\begin{tikzpicture}
	\begin{pgfonlayer}{nodelayer}
		\node [style=geq] (0) at (0.75, 0) {};
		\node [style=white] (1) at (1.75, 0) {};
		\node [style=none] (2) at (-0.25, 0) {};
	\end{pgfonlayer}
	\begin{pgfonlayer}{edgelayer}
		\draw (1) to (0);
		\draw (2.center) to (0);
	\end{pgfonlayer}
\end{tikzpicture}
}
\tikzset{x=1em, y=1.5ex}
, 
\tikzset{x=1em, y=2.1ex}
\begin{tikzpicture}
	\begin{pgfonlayer}{nodelayer}
		\node [style=leq] (0) at (0.75, 0) {};
		\node [style=white] (1) at (1.75, 0) {};
		\node [style=none] (2) at (-0.25, 0) {};
	\end{pgfonlayer}
	\begin{pgfonlayer}{edgelayer}
		\draw (1) to (0);
		\draw (2.center) to (0);
	\end{pgfonlayer}
\end{tikzpicture}
}
\tikzset{x=1em, y=1.5ex}
\big\}$
    such that:
    \[
        \circleeffect{d} = 
\tikzset{x=1em, y=2.1ex}
\InputIfFileExists{poly-normal-form.tikz}{}{\input{./tikz/poly-normal-form.tikz}}
\tikzset{x=1em, y=1.5ex}

    \]

    Where the $d_i$ are minimal in the following sense:
    fixing the set of hyperplanes $H_i$,
    we consider all choices of $d_i$ that give $d$ when composed as above.
    We then require the $d_i$ in the normal form to be minimal (wrt the order of $\IHpoly$) among those.
    We call the set of the $d_i$ a \textit{valuation} for $d$ relative to the hyperplanes $H_i$.
\end{definition}

\begin{definition}
    We say that a morphism $D$ of $\freeUProp\Sigma_\geq^+$ is in pl normal form if it is written as a non-empty union of
    diagrams $d_i$ each in the language of $\freeProp\Sigma_\geq^+$ (i.e. without unions),
    the $d_i$ are in the normal form defined in \cref{def:poly-nf},
    and all the normal forms use the same set of hyperplanes.
\end{definition}

\begin{lemma}\label{thm:poly-nf}
    Every $d\from n \to 0$ in $\freeProp\Sigma_\geq^+$ has a polyhedral normal form.
\end{lemma}
\begin{proof}
    The normal form from \cite[Theorem 14]{bonchiDiagrammaticPolyhedralAlgebra2021} already has the right shape.
    We only need to find a minimal valuation.
    Observe that the intersection of two valuations for $d$ is again a valuation for $d$:
    let $v$ and $v'$ be two valuations for $d$ relative to the hyperplanes $H_i$.
    If we write $\scalar{A} := 
\tikzset{x=1em, y=2.1ex}
\InputIfFileExists{intersect-hyperplanes.tikz}{}{\input{./tikz/intersect-hyperplanes.tikz}}
\tikzset{x=1em, y=1.5ex}
$
    then $
\tikzset{x=1em, y=2.1ex}
\InputIfFileExists{valuation-intersect.tikz}{}{\input{./tikz/valuation-intersect.tikz}}
\tikzset{x=1em, y=1.5ex}
 = 
\tikzset{x=1em, y=2.1ex}
\InputIfFileExists{valuation-intersect-1.tikz}{}{\input{./tikz/valuation-intersect-1.tikz}}
\tikzset{x=1em, y=1.5ex}
 = 
\tikzset{x=1em, y=2.1ex}
\InputIfFileExists{valuation-intersect-2.tikz}{}{\input{./tikz/valuation-intersect-2.tikz}}
\tikzset{x=1em, y=1.5ex}
 = \circleeffect{d}$

    Therefore $v \cap v'$ is again a valuation for $d$.
    Since there are finitely many valuations, we construct the minimal one by intersecting them all.
    \qed
\end{proof}

\begin{lemma}\label{lem:add-hyperplane}
    If a morphism $D$ of $\freeUProp\Sigma_\geq^+$ is in pl normal form and $H$ is a hyperplane,
    there exists $C$ in pl normal form such that $D \pleq C$ and
    $\Hyperplanes(C) = \Hyperplanes(D) \cup \{H\}$.
\end{lemma}
\begin{proof}
    We write the normal form of $D$ as $D = \bigcup_i d_i$.
    Define $C$ to be the following morphism:
    \[
        \circleeffect{C} = \bigcup_i 
\tikzset{x=1em, y=2.1ex}
\InputIfFileExists{Y+.tikz}{}{\input{./tikz/Y+.tikz}}
\tikzset{x=1em, y=1.5ex}
 \cup\; \bigcup_i 
\tikzset{x=1em, y=2.1ex}
\InputIfFileExists{Y-.tikz}{}{\input{./tikz/Y-.tikz}}
\tikzset{x=1em, y=1.5ex}

    \]

    We transform $C$ into $C'$ by reducing all the terms in the union to
    polyhedral normal form.
    This makes $C'$ be in pl normal form.
    Since we add the same hyperplane $H$ to all $d_i$,
    $\Hyperplanes(C') = \Hyperplanes(D) \cup \{H\}$.

    Moreover:
    \begin{equation*}
    \circleeffect{C'}
    = \circleeffect{C}
        = 
\tikzset{x=1em, y=2.1ex}
\InputIfFileExists{Y+UY-.tikz}{}{\input{./tikz/Y+UY-.tikz}}
\tikzset{x=1em, y=1.5ex}

        \myeq{total}  
\tikzset{x=1em, y=2.1ex}
\InputIfFileExists{Y+UY-tot.tikz}{}{\input{./tikz/Y+UY-tot.tikz}}
\tikzset{x=1em, y=1.5ex}

        = \circleeffect{D}
    \end{equation*}
    \qed
\end{proof}

\begin{theorem}\label{lem:pl-nf}
    Every morphism of $\freeUProp\Sigma_\geq^+$ has a pl normal form.
\end{theorem}
\begin{proof}
    Let $D$ be a $n \to 0$ morphism of $\freeUProp\Sigma_\geq^+$.
    First using distributivity of the union over sequential and parallel composition,
    we move all the uses of the union to the top-level.

    Thus $D$ is written $\bigcup_i d_i$ where each $d_i$ doesn't use the union, i.e. is
    in the language of $\freeProp\Sigma_\geq^+$.
    We then rewrite each $d_i$ into polyhedral normal form using \cref{thm:poly-nf}.

    Each $d_i$ is thus also individually in \emph{pl normal form},
    so we can use \cref{lem:add-hyperplane} to add to each $d_i$ all the hyperplanes of the other $d_j$.
    For each $i$ we get a new diagram $d'_i \pleq d_i$ in pl normal form,
    and all the $d'_i$ use the same set of hyperplanes.
    So $\bigcup_i d'_i$ is a pl normal form for $D$.
    \qed
\end{proof}

Before we can prove completeness, we need a final notion:
the interior of a polyhedral relation,
which is the set of its points that don't touch any of its faces.

\begin{definition}
    Let $d$ be morphism in polyhedral normal form.
    We define $\Interior(d)$ to be the set of points $x \in \sem{d}$ for which $H_i(x) \neq 0$ when
    $\circleeffect{d_i} \neq \Wcounit$.
    In other words, $H_i(x)$ is nonzero for all hyperplanes where it can be nonzero without $x$ leaving $\sem{d}$.
\end{definition}

Note that we define $\Interior$ only on polyhedral normal form diagrams.
$\Interior$ appears to be representation-independent at least when $\field = \R$,
but we won't try to prove it in the general case
as we don't need this here.

\begin{remark}
    This is not the usual topological notion of interior.
    In particular, this notion is independent from the dimension of the surrounding space:
    a polyhedron of dimension $0 < k < n$ within $\R^n$ has an empty topological interior
    but a nonempty $\Interior$, as we'll see in the next theorem.
    $\Interior(d)$ instead coincides with the interior of $d$ with the topology of
    the smallest containing affine space.
\end{remark}

\begin{lemma}\label{thm:int-nonempty}
    Let $d$ be a diagram in polyhedral normal form. If $\sem{d}$ is nonempty, then $\Interior(d)$ is nonempty.
\end{lemma}
\begin{proof}
    First, write $d$ in polyhedral normal form $d = \bigcap_i 
\tikzset{x=1em, y=2.1ex}
\begin{tikzpicture}
	\begin{pgfonlayer}{nodelayer}
		\node [style=none] (2) at (-1.5, 0) {};
		\node [style=reg] (9) at (-0.25, 0) {$H_i$};
		\node [style=place] (28) at (1.5, 0) {$d_i$};
	\end{pgfonlayer}
	\begin{pgfonlayer}{edgelayer}
		\draw (9) to (28);
		\draw (2.center) to (9);
	\end{pgfonlayer}
\end{tikzpicture}
}
\tikzset{x=1em, y=1.5ex}
$.
    Up to negating some of the $H_i$, we can assume that
    none of the $\circleeffect{d_i}$ are $
\tikzset{x=1em, y=2.1ex}
\begin{tikzpicture}
	\begin{pgfonlayer}{nodelayer}
		\node [style=leq] (0) at (0.75, 0) {};
		\node [style=white] (1) at (1.75, 0) {};
		\node [style=none] (2) at (-0.25, 0) {};
	\end{pgfonlayer}
	\begin{pgfonlayer}{edgelayer}
		\draw (1) to (0);
		\draw (2.center) to (0);
	\end{pgfonlayer}
\end{tikzpicture}
}
\tikzset{x=1em, y=1.5ex}
$.
    If $\forall i. \circleeffect{d_i} = \Wcounit$,
    then by definition $\Interior(d) = \sem{d}$ which is nonempty so we're done.
    Assume then that $\circleeffect{d_i} = 
\tikzset{x=1em, y=2.1ex}
\begin{tikzpicture}
	\begin{pgfonlayer}{nodelayer}
		\node [style=geq] (0) at (0.75, 0) {};
		\node [style=white] (1) at (1.75, 0) {};
		\node [style=none] (2) at (-0.25, 0) {};
	\end{pgfonlayer}
	\begin{pgfonlayer}{edgelayer}
		\draw (1) to (0);
		\draw (2.center) to (0);
	\end{pgfonlayer}
\end{tikzpicture}
}
\tikzset{x=1em, y=1.5ex}
$ for at least some $i$.
    For each such $i$,
    by minimality of the $d_i$ in the normal form
    there must be a $x_i \in \sem{d}$ such that $H_i(x_i) > 0$.
    We pick such an $x_i$ for each $i$, and define $x$ to be their average.
    By convexity, $x \in \sem{d}$.
    Then for each $i$ either $\circleeffect{d_i} = \Wcounit$
    or $H_i(x) \geq H_i(x_i)/\lvert\{x_j\}_j\rvert > 0$,
    hence $x \in \Interior(d)$.
    \qed
\end{proof}

\begin{theorem}[Completeness]\label{thm:completeness}
    $\sem{D} \subseteq \sem{C} \implies D \plleq C$
\end{theorem}
\begin{proof}
    Using \cref{thm:map-state} we can without loss of generality assume
    that $D$ and $C$ have $n$ inputs and $0$ outputs.
    Using \cref{lem:pl-nf}, we reduce $D$ and $C$ into pl normal form.
    Using \cref{lem:add-hyperplane}, we add each others' hyperplanes to $D$ and $C$ so that
    they both use the exact same set.
    So $D = \bigcup_i d_i$ and $C = \bigcup_i c_i$, where the $d_i$ and $c_i$ are in polyhedral normal form and use a same set of hyperplanes $\{H_i\}_i$.
    Pick one of the $d_i$ in $D$.

    If $d_i$ is the empty polyhedron, we have $\sem{d_i} = \emptyset \subseteq \sem{c_0}$,
    so by completeness of $\IHpoly$ we get $d_i \polyleq c_0$.
    Thus $d_i \plleq c_0 \plleq C$.

    Otherwise $d_i$ is nonempty, and using \cref{thm:int-nonempty} we pick $x \in \Interior(d_i)$.
    Then:
    \[
    x \in \Interior(d_i) \subseteq \sem{d_i} \subseteq \sem{D} \subseteq \sem{C} = \sem{\bigcup_j c_j} = \bigcup_j \sem{c_j}
    \]
    Thus there is a $j$ such that $x \in \sem{c_j}$.
    Now pick a $k$.
    If $\circleeffect{d_{ik}} = \Wcounit$, then $\circleeffect{d_{ik}} \plleq \;\circleeffect{c_{jk}}$
    regardless of $\circleeffect{c_{jk}}$.
    If $\circleeffect{d_{ik}} = 
\tikzset{x=1em, y=2.1ex}
\begin{tikzpicture}
	\begin{pgfonlayer}{nodelayer}
		\node [style=geq] (0) at (0.75, 0) {};
		\node [style=white] (1) at (1.75, 0) {};
		\node [style=none] (2) at (-0.25, 0) {};
	\end{pgfonlayer}
	\begin{pgfonlayer}{edgelayer}
		\draw (1) to (0);
		\draw (2.center) to (0);
	\end{pgfonlayer}
\end{tikzpicture}
}
\tikzset{x=1em, y=1.5ex}
$,
    then by definition of $\Interior(d_i)$, we have $H_k(x) > 0$.
    Since moreover $x \in \sem{c_j}$, $\circleeffect{c_{jk}}$ must be $
\tikzset{x=1em, y=2.1ex}
\begin{tikzpicture}
	\begin{pgfonlayer}{nodelayer}
		\node [style=geq] (0) at (0.75, 0) {};
		\node [style=white] (1) at (1.75, 0) {};
		\node [style=none] (2) at (-0.25, 0) {};
	\end{pgfonlayer}
	\begin{pgfonlayer}{edgelayer}
		\draw (1) to (0);
		\draw (2.center) to (0);
	\end{pgfonlayer}
\end{tikzpicture}
}
\tikzset{x=1em, y=1.5ex}
$.
    If $\circleeffect{d_{ik}} = 
\tikzset{x=1em, y=2.1ex}
\begin{tikzpicture}
	\begin{pgfonlayer}{nodelayer}
		\node [style=leq] (0) at (0.75, 0) {};
		\node [style=white] (1) at (1.75, 0) {};
		\node [style=none] (2) at (-0.25, 0) {};
	\end{pgfonlayer}
	\begin{pgfonlayer}{edgelayer}
		\draw (1) to (0);
		\draw (2.center) to (0);
	\end{pgfonlayer}
\end{tikzpicture}
}
\tikzset{x=1em, y=1.5ex}
$,
    similarly $\circleeffect{c_{jk}}$ must be $
\tikzset{x=1em, y=2.1ex}
\begin{tikzpicture}
	\begin{pgfonlayer}{nodelayer}
		\node [style=leq] (0) at (0.75, 0) {};
		\node [style=white] (1) at (1.75, 0) {};
		\node [style=none] (2) at (-0.25, 0) {};
	\end{pgfonlayer}
	\begin{pgfonlayer}{edgelayer}
		\draw (1) to (0);
		\draw (2.center) to (0);
	\end{pgfonlayer}
\end{tikzpicture}
}
\tikzset{x=1em, y=1.5ex}
$.
    In all three cases, $\circleeffect{d_{ik}} \plleq \;\circleeffect{c_{jk}}$.
    This is the case for every $k$, so:
    \[\circleeffect{d_i}\; = \; 
\tikzset{x=1em, y=2.1ex}
\InputIfFileExists{X-i-hyperplanes.tikz}{}{\input{./tikz/X-i-hyperplanes.tikz}}
\tikzset{x=1em, y=1.5ex}
\;\subseteq \; 
\tikzset{x=1em, y=2.1ex}
\InputIfFileExists{Y-j-hyperplanes.tikz}{}{\input{./tikz/Y-j-hyperplanes.tikz}}
\tikzset{x=1em, y=1.5ex}
 = \; \circleeffect{c_j} \;\subseteq \; \circleeffect{C}\]
    Finally, since we have $d_i \plleq C$ for all $i$,
    we derive $D = \bigcup_i d_i \plleq C$.
    \qed
\end{proof}

\section{Generating Piecewise-Linear Relations}\label{sec:alternative-syntax}


Piecewise-linear subsets of vector spaces give us a rather wide semantic space to explore. 
One might suspect that there exist useful structured relations 
that live strictly between the linear and piecewise-linear worlds. 

Formally, we're interested in finding sub-props of $\RelX{\field}$ that
contains not only linear or polyhedral relations, but some selected non-convex relations
that would be useful for particular applications.
It turns out that for many sensible choices,
the resulting image will coincide with pl relations---a somewhat surprising fact.
Note that we are interested in generating sub-props of $\RelX{\field}$ here, not $\cup$-props, since the $\cup$-prop generated by the image of $\freeUProp\Sigma_\geq^+$ under $\sem{\cdot}$ already contains all pl relations. 

We will go through a few natural choices, each time defining them as a term of $\freeUProp\Sigma_\geq^+$, a shortcut which makes reasoning about them much easier than with their set-theoretic semantics. Of course, their semantics in $\RelX{\field}$ can be recovered via $\sem{\cdot}$.

\subsection{The $n$-Fold Union Generators}

We first show that the main difference between polyhedral and pl relations
---the unions---can be bridged.
Indeed, it is not obvious that we can build arbitrary unions of diagrams
without having access to the syntax of a SMLT.
For this we introduce a family of diagrams we call
the $n$-fold union generators, defined for a given 
$n$ as:
\[
\tikzset{x=1em, y=2.1ex}
\InputIfFileExists{union-n.tikz}{}{\input{./tikz/union-n.tikz}}
\tikzset{x=1em, y=1.5ex}
 \; :=\; 
\tikzset{x=1em, y=2.1ex}
\InputIfFileExists{union-left.tikz}{}{\input{./tikz/union-left.tikz}}
\tikzset{x=1em, y=1.5ex}
 \;\cup \;
\tikzset{x=1em, y=2.1ex}
\InputIfFileExists{union-right.tikz}{}{\input{./tikz/union-right.tikz}}
\tikzset{x=1em, y=1.5ex}
\]

\noindent This family of generators suffices to reproduce the behaviour of the syntactic union:
\begin{theorem}
    The image of the free prop generated by $\Sigma_\geq^+$ and the $n$-fold union generators for all $n$ is the prop of pl relations.
\end{theorem}
\begin{proof}
    If \circleeffect{C} and \circleeffect{D} are non-empty $n \to 0$ diagrams,
    \[
\tikzset{x=1em, y=2.1ex}
\InputIfFileExists{AUB.tikz}{}{\input{./tikz/AUB.tikz}}
\tikzset{x=1em, y=1.5ex}
 \;=\; 
\tikzset{x=1em, y=2.1ex}
\InputIfFileExists{AUB-left.tikz}{}{\input{./tikz/AUB-left.tikz}}
\tikzset{x=1em, y=1.5ex}
 \;\cup 
\tikzset{x=1em, y=2.1ex}
\InputIfFileExists{AUB-right.tikz}{}{\input{./tikz/AUB-right.tikz}}
\tikzset{x=1em, y=1.5ex}
 = \circleeffect{C}\, \cup \circleeffect{D}\]

    Since every pl relation can be written as
    a finite union of diagrams in $\freeProp\Sigma_\geq^+$,
    and we can easily avoid diagrams denoting the empty relation,
    this generates all of pl relations.
    \qed
\end{proof}

This means that we didn't formally need to introduce the notion of a SMLT after all:
we could have defined an equivalent SMIT by adding these generators.
However, this is for most purposes a much less convenient syntax,
and the corresponding equational theory would be more difficult to calculate with.
This is also the case for the examples that follow.

\subsection{The Simplest Non-Convex Diagram}

The following is one of the simplest diagrams that captures a non-convex relation:
\[
\tikzset{x=1em, y=2.1ex}
\begin{tikzpicture}
	\begin{pgfonlayer}{nodelayer}
		\node [style=none] (13) at (1.5, 0) {};
		\node [style=none] (17) at (-1.5, 0) {};
		\node [style=basic rounded box] (18) at (0, 0) {$+$};
	\end{pgfonlayer}
	\begin{pgfonlayer}{edgelayer}
		\draw (17.center) to (18);
		\draw (18) to (13.center);
	\end{pgfonlayer}
\end{tikzpicture}
}
\tikzset{x=1em, y=1.5ex}
 \;:=\; \Wcounit \,\; \Bunit \cup \Bcounit \,\; \Wunit\]
\noindent It is named after its semantics: the union of the $x$ and $y$ axes in the plane,
corresponding to the simple equation $x = 0 \lor y = 0$.
Despite its simplicity, it suffices to generate all of pl relations.

\begin{theorem}
    The image of the free prop generated by $\Sigma_\geq^+$ and $
\tikzset{x=1em, y=2.1ex}
\begin{tikzpicture}
	\begin{pgfonlayer}{nodelayer}
		\node [style=none] (13) at (1.5, 0) {};
		\node [style=none] (17) at (-1.5, 0) {};
		\node [style=basic rounded box] (18) at (0, 0) {$+$};
	\end{pgfonlayer}
	\begin{pgfonlayer}{edgelayer}
		\draw (17.center) to (18);
		\draw (18) to (13.center);
	\end{pgfonlayer}
\end{tikzpicture}
}
\tikzset{x=1em, y=1.5ex}
$ is the prop of pl relations. 
\end{theorem}
\begin{proof}
    Define $dup : 1 \to 2$:\;
    $
\tikzset{x=1em, y=2.1ex}
\InputIfFileExists{dup.tikz}{}{\input{./tikz/dup.tikz}}
\tikzset{x=1em, y=1.5ex}
\; := \;
\tikzset{x=1em, y=2.1ex}
\InputIfFileExists{copy-abs.tikz}{}{\input{./tikz/copy-abs.tikz}}
\tikzset{x=1em, y=1.5ex}
$
\vspace{2pt}

\noindent This diagram has the interesting property of duplicating black and white units:
    \[
\tikzset{x=1em, y=2.1ex}
\InputIfFileExists{wunit-dup.tikz}{}{\input{./tikz/wunit-dup.tikz}}
\tikzset{x=1em, y=1.5ex}
 \; = \; 
\tikzset{x=1em, y=2.1ex}
\begin{tikzpicture}
	\begin{pgfonlayer}{nodelayer}
		\node [style=none] (14) at (0.5, 0.5) {};
		\node [style=white] (24) at (-1, 0.5) {};
		\node [style=none] (25) at (0.5, -0.5) {};
		\node [style=white] (26) at (-1, -0.5) {};
	\end{pgfonlayer}
	\begin{pgfonlayer}{edgelayer}
		\draw (14.center) to (24);
		\draw (25.center) to (26);
	\end{pgfonlayer}
\end{tikzpicture}
}
\tikzset{x=1em, y=1.5ex}
\qquad\quad\qquad 
\tikzset{x=1em, y=2.1ex}
\InputIfFileExists{bunit-dup.tikz}{}{\input{./tikz/bunit-dup.tikz}}
\tikzset{x=1em, y=1.5ex}
 \; = \; 
\tikzset{x=1em, y=2.1ex}
\begin{tikzpicture}
	\begin{pgfonlayer}{nodelayer}
		\node [style=none] (14) at (0.5, 0.5) {};
		\node [style=black] (24) at (-1, 0.5) {};
		\node [style=none] (25) at (0.5, -0.5) {};
		\node [style=black] (26) at (-1, -0.5) {};
	\end{pgfonlayer}
	\begin{pgfonlayer}{edgelayer}
		\draw (14.center) to (24);
		\draw (25.center) to (26);
	\end{pgfonlayer}
\end{tikzpicture}
}
\tikzset{x=1em, y=1.5ex}
\]

    \noindent We can chain it to build $
\tikzset{x=1em, y=2.1ex}
\InputIfFileExists{dup-n.tikz}{}{\input{./tikz/dup-n.tikz}}
\tikzset{x=1em, y=1.5ex}
 := 
\tikzset{x=1em, y=2.1ex}
\InputIfFileExists{dup-n-def.tikz}{}{\input{./tikz/dup-n-def.tikz}}
\tikzset{x=1em, y=1.5ex}
$ for any $n$.

    \noindent Then, let $
\tikzset{x=1em, y=2.1ex}
\InputIfFileExists{+n.tikz}{}{\input{./tikz/+n.tikz}}
\tikzset{x=1em, y=1.5ex}
 \;:=\; 
\tikzset{x=1em, y=2.1ex}
\InputIfFileExists{co-dupn-+-dupn.tikz}{}{\input{./tikz/co-dupn-+-dupn.tikz}}
\tikzset{x=1em, y=1.5ex}
\; = \;
\tikzset{x=1em, y=2.1ex}
\InputIfFileExists{wcounitn-bunitn.tikz}{}{\input{./tikz/wcounitn-bunitn.tikz}}
\tikzset{x=1em, y=1.5ex}
\,\cup 
\tikzset{x=1em, y=2.1ex}
\InputIfFileExists{bcounitn-wunitn.tikz}{}{\input{./tikz/bcounitn-wunitn.tikz}}
\tikzset{x=1em, y=1.5ex}
$

    \noindent This allows us to build:
    \begin{align*}
    
\tikzset{x=1em, y=2.1ex}
\InputIfFileExists{union-from-+.tikz}{}{\input{./tikz/union-from-+.tikz}}
\tikzset{x=1em, y=1.5ex}
 \;&=\;
\tikzset{x=1em, y=2.1ex}
\InputIfFileExists{union-from-+-left.tikz}{}{\input{./tikz/union-from-+-left.tikz}}
\tikzset{x=1em, y=1.5ex}
\,\cup 
\tikzset{x=1em, y=2.1ex}
\InputIfFileExists{union-from-+-right.tikz}{}{\input{./tikz/union-from-+-right.tikz}}
\tikzset{x=1em, y=1.5ex}
 \\&=\; 
\tikzset{x=1em, y=2.1ex}
\InputIfFileExists{union-left.tikz}{}{\input{./tikz/union-left.tikz}}
\tikzset{x=1em, y=1.5ex}
 \;\cup \;
\tikzset{x=1em, y=2.1ex}
\InputIfFileExists{union-right.tikz}{}{\input{./tikz/union-right.tikz}}
\tikzset{x=1em, y=1.5ex}
 \;= \;
\tikzset{x=1em, y=2.1ex}
\InputIfFileExists{union-n.tikz}{}{\input{./tikz/union-n.tikz}}
\tikzset{x=1em, y=1.5ex}

    \end{align*}\qed

\end{proof}

\subsection{The Semantics of a Diode}

\begin{minipage}{0.7\textwidth}
Most basic electrical circuit components can be modelled with an affine semantics.
The first exception is the (ideal) diode:
the idealised current-voltage semantics across a diode
is that the current can be negative and the voltage difference positive
but not both at the same time.
\end{minipage}
\begin{minipage}{0.3\textwidth}
\[
\tikzset{x=1em, y=2.1ex}
\begin{tikzpicture}[scale=0.6]
    \node[anchor=west] at (5,0) {$I$};
    \node[anchor=south] at (0,5) {$U$};
    \draw[->] (-5,0) -- (5,0);
    \draw[->] (0,-5) -- (0,5);
    \draw[very thick,color=blue] (0,0) -- (0,4.5);
    \draw[very thick,color=blue] (-4.5,0) -- (0,0);
\end{tikzpicture}
}
\tikzset{x=1em, y=1.5ex}
\]
\end{minipage}
On a graph, the allowed (current, voltage difference) pairs are depicted above.
Not only is this not affine, it is not even convex.
The corresponding diagram, $
\tikzset{x=1em, y=2.1ex}
\InputIfFileExists{leq-wcounit-wunit.tikz}{}{\input{./tikz/leq-wcounit-wunit.tikz}}
\tikzset{x=1em, y=1.5ex}
\, \cup 
\tikzset{x=1em, y=2.1ex}
\InputIfFileExists{wcounit-wunit-leq.tikz}{}{\input{./tikz/wcounit-wunit-leq.tikz}}
\tikzset{x=1em, y=1.5ex}
$,
is outside of both affine and polyhedral algebra.

We will see how to model electrical circuits with diodes in more detail
in the next section.
We will focus here on the following fact:
adding a generator with this semantics is once again enough
to recover all pl relations.
In fact we can even build the $\geq$ relation from the diode,
so we can start from affine algebra (without requiring the generality of polyhedral algebra).

For convenience, we define a new generator whose semantics is the mirror image of the diode's graph: 
\[\smallubox{\mathsf{L}}{}{}\;:=\;
\tikzset{x=1em, y=2.1ex}
\InputIfFileExists{geq-wcounit-wunit.tikz}{}{\input{./tikz/geq-wcounit-wunit.tikz}}
\tikzset{x=1em, y=1.5ex}
\, \cup 
\tikzset{x=1em, y=2.1ex}
\InputIfFileExists{wcounit-wunit-leq.tikz}{}{\input{./tikz/wcounit-wunit-leq.tikz}}
\tikzset{x=1em, y=1.5ex}
\]
Recall that $\Sigma^+$ is $\Sigma_\geq^+$ without $
\tikzset{x=1em, y=2.1ex}
\begin{tikzpicture}
	\begin{pgfonlayer}{nodelayer}
		\node [style=geq] (0) at (0, 0) {};
		\node [style=none] (2) at (-1.25, 0) {};
		\node [style=none] (3) at (1.25, 0) {};
	\end{pgfonlayer}
	\begin{pgfonlayer}{edgelayer}
		\draw [in=0, out=180] (0) to (2.center);
		\draw (3.center) to (0);
	\end{pgfonlayer}
\end{tikzpicture}
}
\tikzset{x=1em, y=1.5ex}
$.
\begin{theorem}\label{thm:diode-is-enough}
    The image of the free prop generated by $\Sigma^+$ and \textsf{L} is the prop of pl relations. 
\end{theorem}
\begin{proof}
    First, we can construct the $\geq$ generator from \textsf{L}:
    \[
\tikzset{x=1em, y=2.1ex}
\InputIfFileExists{wcomult-L-wcounit.tikz}{}{\input{./tikz/wcomult-L-wcounit.tikz}}
\tikzset{x=1em, y=1.5ex}
\,=\, 
\tikzset{x=1em, y=2.1ex}
\InputIfFileExists{wcomult-geq-wcounit-wbone.tikz}{}{\input{./tikz/wcomult-geq-wcounit-wbone.tikz}}
\tikzset{x=1em, y=1.5ex}
\cup 
\tikzset{x=1em, y=2.1ex}
\InputIfFileExists{wcomult-wcounit-wunit-geq-wcounit.tikz}{}{\input{./tikz/wcomult-wcounit-wunit-geq-wcounit.tikz}}
\tikzset{x=1em, y=1.5ex}
 = \;
\tikzset{x=1em, y=2.1ex}
\begin{tikzpicture}
	\begin{pgfonlayer}{nodelayer}
		\node [style=geq] (0) at (0, 0) {};
		\node [style=none] (2) at (-1.25, 0) {};
		\node [style=none] (3) at (1.25, 0) {};
	\end{pgfonlayer}
	\begin{pgfonlayer}{edgelayer}
		\draw [in=0, out=180] (0) to (2.center);
		\draw (3.center) to (0);
	\end{pgfonlayer}
\end{tikzpicture}
}
\tikzset{x=1em, y=1.5ex}
 \cup \idone = 
\tikzset{x=1em, y=2.1ex}
\begin{tikzpicture}
	\begin{pgfonlayer}{nodelayer}
		\node [style=geq] (0) at (0, 0) {};
		\node [style=none] (2) at (-1.25, 0) {};
		\node [style=none] (3) at (1.25, 0) {};
	\end{pgfonlayer}
	\begin{pgfonlayer}{edgelayer}
		\draw [in=0, out=180] (0) to (2.center);
		\draw (3.center) to (0);
	\end{pgfonlayer}
\end{tikzpicture}
}
\tikzset{x=1em, y=1.5ex}
\]

    So we generate all polyhedral relations.
    Then we can also recover the $+$ generator from the previous section,
    which is enough to generate all of pl relations:
    \begin{align*}
    
\tikzset{x=1em, y=2.1ex}
\InputIfFileExists{bw-conv-L-minus-L.tikz}{}{\input{./tikz/bw-conv-L-minus-L.tikz}}
\tikzset{x=1em, y=1.5ex}
 &= 
\tikzset{x=1em, y=2.1ex}
\InputIfFileExists{L-Lminus-1.tikz}{}{\input{./tikz/L-Lminus-1.tikz}}
\tikzset{x=1em, y=1.5ex}
\cup
\tikzset{x=1em, y=2.1ex}
\InputIfFileExists{L-Lminus-2.tikz}{}{\input{./tikz/L-Lminus-2.tikz}}
\tikzset{x=1em, y=1.5ex}
\\
    &\quad\cup
\tikzset{x=1em, y=2.1ex}
\InputIfFileExists{L-Lminus-3.tikz}{}{\input{./tikz/L-Lminus-3.tikz}}
\tikzset{x=1em, y=1.5ex}
\cup
\tikzset{x=1em, y=2.1ex}
\InputIfFileExists{L-Lminus-4.tikz}{}{\input{./tikz/L-Lminus-4.tikz}}
\tikzset{x=1em, y=1.5ex}
\\
    &= 
\tikzset{x=1em, y=2.1ex}
\InputIfFileExists{geq-wcounit-wunit.tikz}{}{\input{./tikz/geq-wcounit-wunit.tikz}}
\tikzset{x=1em, y=1.5ex}
\cup 
\tikzset{x=1em, y=2.1ex}
\InputIfFileExists{wcounit-wunit-geq.tikz}{}{\input{./tikz/wcounit-wunit-geq.tikz}}
\tikzset{x=1em, y=1.5ex}
\cup 
\tikzset{x=1em, y=2.1ex}
\InputIfFileExists{wcounit-wunit-leq.tikz}{}{\input{./tikz/wcounit-wunit-leq.tikz}}
\tikzset{x=1em, y=1.5ex}
\cup 
\tikzset{x=1em, y=2.1ex}
\begin{tikzpicture}
	\begin{pgfonlayer}{nodelayer}
		\node [style=black] (14) at (-0.75, 0) {};
		\node [style=none] (27) at (0.5, 0) {};
		\node [style=white] (29) at (-1.75, 0) {};
		\node [style=none] (31) at (-3, 0) {};
	\end{pgfonlayer}
	\begin{pgfonlayer}{edgelayer}
		\draw (14) to (27.center);
		\draw (29) to (31.center);
	\end{pgfonlayer}
\end{tikzpicture}
}
\tikzset{x=1em, y=1.5ex}
\\
    &= 
\tikzset{x=1em, y=2.1ex}
\InputIfFileExists{geq-wcounit-wunit.tikz}{}{\input{./tikz/geq-wcounit-wunit.tikz}}
\tikzset{x=1em, y=1.5ex}
\cup 
\tikzset{x=1em, y=2.1ex}
\begin{tikzpicture}
	\begin{pgfonlayer}{nodelayer}
		\node [style=black] (14) at (-0.75, 0) {};
		\node [style=none] (27) at (0.5, 0) {};
		\node [style=white] (29) at (-1.75, 0) {};
		\node [style=none] (31) at (-3, 0) {};
	\end{pgfonlayer}
	\begin{pgfonlayer}{edgelayer}
		\draw (14) to (27.center);
		\draw (29) to (31.center);
	\end{pgfonlayer}
\end{tikzpicture}
}
\tikzset{x=1em, y=1.5ex}
 =: \smallubox{\vdash}
    \end{align*}
    \begin{align*}
    
\tikzset{x=1em, y=2.1ex}
\InputIfFileExists{wb-conv-vdash-minusvdash.tikz}{}{\input{./tikz/wb-conv-vdash-minusvdash.tikz}}
\tikzset{x=1em, y=1.5ex}
 &= 
\tikzset{x=1em, y=2.1ex}
\InputIfFileExists{vdash-minusvdash-1.tikz}{}{\input{./tikz/vdash-minusvdash-1.tikz}}
\tikzset{x=1em, y=1.5ex}
 \cup 
\tikzset{x=1em, y=2.1ex}
\InputIfFileExists{vdash-minusvdash-2.tikz}{}{\input{./tikz/vdash-minusvdash-2.tikz}}
\tikzset{x=1em, y=1.5ex}
\\
    &\quad \cup 
\tikzset{x=1em, y=2.1ex}
\InputIfFileExists{vdash-minusvdash-3.tikz}{}{\input{./tikz/vdash-minusvdash-3.tikz}}
\tikzset{x=1em, y=1.5ex}
	\cup 
\tikzset{x=1em, y=2.1ex}
\InputIfFileExists{vdash-minusvdash-4.tikz}{}{\input{./tikz/vdash-minusvdash-4.tikz}}
\tikzset{x=1em, y=1.5ex}
\\
    &= 
\tikzset{x=1em, y=2.1ex}
\begin{tikzpicture}
	\begin{pgfonlayer}{nodelayer}
		\node [style=white] (14) at (-0.75, 0) {};
		\node [style=none] (27) at (0.5, 0) {};
		\node [style=black] (29) at (-1.75, 0) {};
		\node [style=none] (31) at (-3, 0) {};
	\end{pgfonlayer}
	\begin{pgfonlayer}{edgelayer}
		\draw (14) to (27.center);
		\draw (29) to (31.center);
	\end{pgfonlayer}
\end{tikzpicture}
}
\tikzset{x=1em, y=1.5ex}
\cup 
\tikzset{x=1em, y=2.1ex}
\InputIfFileExists{geq-wcounit-wunit.tikz}{}{\input{./tikz/geq-wcounit-wunit.tikz}}
\tikzset{x=1em, y=1.5ex}
 \cup 
\tikzset{x=1em, y=2.1ex}
\InputIfFileExists{leq-wcounit-wunit.tikz}{}{\input{./tikz/leq-wcounit-wunit.tikz}}
\tikzset{x=1em, y=1.5ex}
\cup 
\tikzset{x=1em, y=2.1ex}
\begin{tikzpicture}
	\begin{pgfonlayer}{nodelayer}
		\node [style=black] (14) at (-0.75, 0) {};
		\node [style=none] (27) at (0.5, 0) {};
		\node [style=white] (29) at (-1.75, 0) {};
		\node [style=none] (31) at (-3, 0) {};
	\end{pgfonlayer}
	\begin{pgfonlayer}{edgelayer}
		\draw (14) to (27.center);
		\draw (29) to (31.center);
	\end{pgfonlayer}
\end{tikzpicture}
}
\tikzset{x=1em, y=1.5ex}
\\
    &= 
\tikzset{x=1em, y=2.1ex}
\begin{tikzpicture}
	\begin{pgfonlayer}{nodelayer}
		\node [style=white] (14) at (-0.75, 0) {};
		\node [style=none] (27) at (0.5, 0) {};
		\node [style=black] (29) at (-1.75, 0) {};
		\node [style=none] (31) at (-3, 0) {};
	\end{pgfonlayer}
	\begin{pgfonlayer}{edgelayer}
		\draw (14) to (27.center);
		\draw (29) to (31.center);
	\end{pgfonlayer}
\end{tikzpicture}
}
\tikzset{x=1em, y=1.5ex}
\cup 
\tikzset{x=1em, y=2.1ex}
\begin{tikzpicture}
	\begin{pgfonlayer}{nodelayer}
		\node [style=black] (14) at (-0.75, 0) {};
		\node [style=none] (27) at (0.5, 0) {};
		\node [style=white] (29) at (-1.75, 0) {};
		\node [style=none] (31) at (-3, 0) {};
	\end{pgfonlayer}
	\begin{pgfonlayer}{edgelayer}
		\draw (14) to (27.center);
		\draw (29) to (31.center);
	\end{pgfonlayer}
\end{tikzpicture}
}
\tikzset{x=1em, y=1.5ex}
 = 
\tikzset{x=1em, y=2.1ex}
\begin{tikzpicture}
	\begin{pgfonlayer}{nodelayer}
		\node [style=none] (13) at (1.5, 0) {};
		\node [style=none] (17) at (-1.5, 0) {};
		\node [style=basic rounded box] (18) at (0, 0) {$+$};
	\end{pgfonlayer}
	\begin{pgfonlayer}{edgelayer}
		\draw (17.center) to (18);
		\draw (18) to (13.center);
	\end{pgfonlayer}
\end{tikzpicture}
}
\tikzset{x=1em, y=1.5ex}

    \end{align*}
    \qed
\end{proof}

\subsection{$max$, \emph{ReLu} and $abs$}

Three of the most basic piecewise-linear functions one might come across are
$abs$, $max$ and $ReLu$.
We define them diagrammatically as follows:
\[
\tikzset{x=1em, y=2.1ex}
\InputIfFileExists{max.tikz}{}{\input{./tikz/max.tikz}}
\tikzset{x=1em, y=1.5ex}
:=
\tikzset{x=1em, y=2.1ex}
\InputIfFileExists{max-def-left.tikz}{}{\input{./tikz/max-def-left.tikz}}
\tikzset{x=1em, y=1.5ex}
\;\cup
\tikzset{x=1em, y=2.1ex}
\InputIfFileExists{max-def-right.tikz}{}{\input{./tikz/max-def-right.tikz}}
\tikzset{x=1em, y=1.5ex}
 \quad \scalar{abs} := 
\tikzset{x=1em, y=2.1ex}
\InputIfFileExists{abs-def.tikz}{}{\input{./tikz/abs-def.tikz}}
\tikzset{x=1em, y=1.5ex}
\]
\[\scalar{ReLu} := 
\tikzset{x=1em, y=2.1ex}
\InputIfFileExists{relu-def.tikz}{}{\input{./tikz/relu-def.tikz}}
\tikzset{x=1em, y=1.5ex}
\]

While the reader will certainly be familiar with the first two, \emph{ReLu} has acquired significant fame as one of the basic building blocks of neural networks. In fact, all neural networks whose activation function is \emph{ReLu} can be represented in GPLA. This opens up the exciting possibility of applying equational reasoning to neural networks, a possibility that we leave for future work. 

Once again, adding either of them to the syntax for affine algebra suffices to construct any pl relation.

\begin{theorem}
The image of the free prop generated by $\Sigma^+$ and any of $max$, $abs$ or $ReLu$ is the prop of pl relations.
\end{theorem}
\begin{proof}
    First, we notice that the three functions are inter-definable.
    $abs$ and $ReLu$ were already defined in terms of $max$,
    and we can complete the cycle:
    \[
        max(x,y) = x + max(0, y-x) = x + ReLu(y-x)
    \]
    \[
        ReLu(x) = max(0, x) = (x + abs(x))/2
    \]
    So we only need to show the result for one of them.
    Let's pick $max$.
    We recover \textsf{L}, which we know suffices by~\cref{thm:diode-is-enough}.
    First $
\tikzset{x=1em, y=2.1ex}
\InputIfFileExists{max-wcounit.tikz}{}{\input{./tikz/max-wcounit.tikz}}
\tikzset{x=1em, y=1.5ex}
 \;=\; 
\tikzset{x=1em, y=2.1ex}
\InputIfFileExists{leq-wcounit+wcounit.tikz}{}{\input{./tikz/leq-wcounit+wcounit.tikz}}
\tikzset{x=1em, y=1.5ex}
\;\; \cup 
\tikzset{x=1em, y=2.1ex}
\InputIfFileExists{wcounit+leq-wcounit.tikz}{}{\input{./tikz/wcounit+leq-wcounit.tikz}}
\tikzset{x=1em, y=1.5ex}
$.

    \noindent Thus $
\tikzset{x=1em, y=2.1ex}
\InputIfFileExists{L-from-max.tikz}{}{\input{./tikz/L-from-max.tikz}}
\tikzset{x=1em, y=1.5ex}
=\; 
\tikzset{x=1em, y=2.1ex}
\InputIfFileExists{geq-wcounit-wunit.tikz}{}{\input{./tikz/geq-wcounit-wunit.tikz}}
\tikzset{x=1em, y=1.5ex}
\cup 
\tikzset{x=1em, y=2.1ex}
\InputIfFileExists{wcounit-wunit-leq.tikz}{}{\input{./tikz/wcounit-wunit-leq.tikz}}
\tikzset{x=1em, y=1.5ex}
 \;=\;\smallubox{\mathsf{L}}$
    \qed
\end{proof}

\begin{remark}
    It is standard that $max$ together with linear maps
    generates all continuous pl \emph{functions}.
    Our result can be seen as a generalization of this fact
    to the relational setting.
\end{remark}


\subsection{Conclusion}

These examples justify the generality of pl relations:
they constitute the minimal extension of polyhedral algebra (and in some cases affine algebra)
that can express any of the very useful relations above. 
This is interesting because pl relations form a nearly universal domain:
they can approximate any smooth manifold over a bounded domain.

Despite our compelling examples,
there could still be interesting props between polyhedral and pl relations.
In particular, determining the prop generated by $\Sigma_\geq^+$
together with $\Wunit \cup \One$
is currently an open problem.

\section{Case Study: Electronic Circuits}\label{sec:electronic}

To illustrate how one would use this theory in a concrete case,
we turn to the study of electronic circuits.
We build on the work done in \cite{boisseauStringDiagrammaticElectrical2021}.
The syntax mimics the usual circuits drawn by electrical engineers,
by generating a free two-colored prop from basic elements and wires.
The blue wires are electrical wires,
and the black wires carry information;
for details see \cite{boisseauStringDiagrammaticElectrical2021}.
\[
    \eleccopy
    ~|~ \eleccocopy
    ~|~ \eleccounit
    ~|~ \elecunit
    ~|~ \tikz[color=blue]{\draw (0,0) to[R, l={\small $R$}] (3,0);}
    ~|~ \tikz[color=blue]{\draw (0,0) to[diode] (3,0);}
    ~|~ \tikz[color=blue]{\draw (0,0) to[vsource, name=V] (3,0); \draw[color=black] (0,1.5) to[in=90, out=0] (V.north);}
    ~|~ \tikz[color=blue]{\draw (0,0) to[I, name=I] (3,0); \draw[color=black] (0,1.5) to[in=90, out=0] (I.north);}
    ~|~ \tikz[color=blue]{\draw (0,0) to[ammeter, name=A] (3,0); \draw[color=black] (A.north) to[in=180, out=90] (3,1.5);}
    ~|~ \tikz[color=blue]{\draw (0,0) to[vmeter, name=V] (3,0); \draw[color=black] (V.north) to[in=180, out=90] (3,1.5);}
\]

The corresponding physical model imposes constraints
between two quantities: current and voltage.
To express this, we map an electrical wire into two GPLA wires,
the top one for voltage and the bottom one for current.
We then give to each generator a semantics in GPLA
that expresses the relevant physical equations.
For example:

\begin{align*}
    \sem{\tikz[color=blue]{\draw (0,0) to[R, l={\small $R$}] (3,0);}}
    &:= \includegraffle{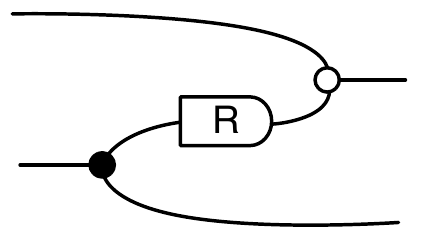},
    &\sem{\eleccopy}
    &:= \includegraffle{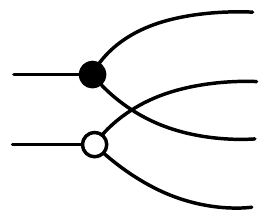},
    &\sem{\tikz[color=blue]{\draw (0,0) to[vsource, name=V] (3,0); \draw[color=black] (0,1.5) to[in=90, out=0] (V.north);}}
    &:= \includegraffle{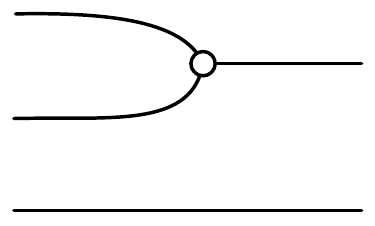}
\end{align*}

The core of this approach is the fact that composition
of constraints in GPLA gives the behaviour of the corresponding composite electrical circuit.
We can thus define the semantics of a whole circuit
compositionally, and get the physically expected result.

So far this follows exactly~\cite{boisseauStringDiagrammaticElectrical2021}.
Our contribution is the ability to express the behaviour of diodes:
\begin{align*}
    \sem{\tikz[color=blue]{\draw (0,0) to[diode] (3,0);}}
    &:= \; 
\tikzset{x=1em, y=2.1ex}
\InputIfFileExists{diode-def-left.tikz}{}{\input{./tikz/diode-def-left.tikz}}
\tikzset{x=1em, y=1.5ex}
\cup
\tikzset{x=1em, y=2.1ex}
\InputIfFileExists{diode-def-right.tikz}{}{\input{./tikz/diode-def-right.tikz}}
\tikzset{x=1em, y=1.5ex}
\\
    &= \;
\tikzset{x=1em, y=2.1ex}
\InputIfFileExists{diode-lagrange-union.tikz}{}{\input{./tikz/diode-lagrange-union.tikz}}
\tikzset{x=1em, y=1.5ex}

    = 
\tikzset{x=1em, y=2.1ex}
\InputIfFileExists{diode-lagrange.tikz}{}{\input{./tikz/diode-lagrange.tikz}}
\tikzset{x=1em, y=1.5ex}

\end{align*}

\begin{remark}
    We cannot include capacitors and inductors,
    because they require semantics in $\IHaff_{\R(x)}$,
    and $\R(x)$ cannot be ordered in a way that would be consistent with the physics.
    Finding diagrammatic semantics that can accommodate both capacitors and diodes is an important open problem.
\end{remark}

This extension allows us to model electronic circuits!
As hinted in the previous section,
diodes by themselves can be used to build many things.
For example, we can model a simple idealized transistor as follows:
\cite[Fig. 59.1]{textbookOfElectricalTechnology}

\[
\begin{tikzpicture}[color=blue]
    \node[pnp, xscale=-1] (T) at (0,0) {};
    \draw (-2,1.7) -- (T.E);
    \draw (-2,-1.7) -- (T.C);
    \draw (T.B) -- (2,0);
\end{tikzpicture}
:=
\begin{tikzpicture}[color=blue]
    \draw (0,1.7) to[ammeter, name=A] (3,1.7);
    \draw (3,1.7) to[diode] (5,1.7);
    \draw (0,-1.7) to[isource, invert, name=I] (5,-1.7);
    \draw[color=black] (A.south) to[in=105, out=-75] (I.north);
    \draw (5,1.7) -- (5,-1.7);
    \node[elecdot] at (5,0) {};
    \draw (5,0) -- (7,0);
\end{tikzpicture}
\]

That said,
it is impractical to prove the equality of two non-trivial electronic circuits explicitly
as the number of alternatives grows exponentially in the number of diodes.
Like in standard mathematical practice,
making this practical will require finding appropriate techniques
and approximations,
which we leave for future work.

\section*{Acknowledgements}
The authors would like to thank the various Twitter and Zulip users
who contributed to the genesis and development of
the theory contained in this paper,
notably Jules Hedges, Cole Comfort and Reid Barton.
Reid Barton in particular contributed significantly to the proof of completeness.

The first author is funded by the \textsc{EPSRC}
under grant~OUCS/GB/1034913.
The second author acknowledges support from \textsc{EPSRC} grant~EP/V002376/1.

\newpage
\appendix

\section{Appendix}

\subsection{Linear, Affine and Polyhedral Relations}

From subsets of the generators in~\eqref{eq:syntax}, we can define four syntactic props and characterise the image of the corresponding semantic functor, in order of increasing expressiveness below, along with their respective axiomatisations.
\begin{description}
\item[Matrices.] \hfill
\begin{itemize}
\item For $\Sigma^{\frac{1}{2}} = \left\{\Bcomult, \Bcounit, \Wmult, \Wunit\right\}$, let $\sem{\cdot}\from \mathsf{P}\Sigma^{\frac{1}{2}} \to \RelX{\field}$ be the prop morphism given on generators by
\begin{equation}\label{eq:mat-semantics}
\begin{gathered} 
\sem{\Bcomult}\df \left\{\left(x,\begin{pmatrix} x\\ x \end{pmatrix}\right) \mid x \in \field \right\} \qquad\quad \sem{\Bcounit}\df \{(x,\bullet) \mid x \in \field \} \\
\sem{\Wmult}\df \left\{\left(\begin{pmatrix} x\\ {y} \end{pmatrix},{x+y}\right) \mid x,{y} \in \field \right\}
\qquad
\sem{\Wunit}\df \{(\bullet, {0})  \}\qquad\qquad\qquad\quad\\
 \sem{\scalar{k}}\df \{(x, k \cdot x) \mid x \in \field \}\quad \text{for } k\in\field.
\end{gathered}
\end{equation}
\item The image of $\mathsf{P}\Sigma^{\frac{1}{2}}$ by $\sem{\cdot}$ is (isomorphic to) the prop whose arrows $n\to m$ are $m\times n$-\emph{matrices} with coefficients in $\field$ (see e.g.~\cite[Section 3]{interactinghopf}). We call this prop $\Mat{\field}$.

\item The theory of bimonoids, with scalars $\scalar{k}$ that distribute over the bimonoid generators axiomatises the prop of matrices (see e.g.~\cite[Proposition 3.7]{interactinghopf}). It can be found in the first block of of \cref{fig:ih-pl}.
\end{itemize} 
\item[Linear relations.]\hfill
\begin{itemize}
\item For $\Sigma = \Sigma^{\frac{1}{2}} \cup \left\{\Bmult, \Bunit, \Wcomult, \Wcounit\right\}$, extend\footnote{We will overload the semantic functor $\sem{\cdot}$ as it poses no problem to do so.} $\sem{\cdot}\from \mathsf{P}\Sigma \to \RelX{\field}$ to the prop morphism defined by~\eqref{eq:mat-semantics} and
\begin{equation}\label{eq:linrel-semantics}
\begin{gathered} 
\sem{\Bmult}\df \left\{\left(\begin{pmatrix} x\\ x \end{pmatrix},x\right) \mid x \in \field \right\} \qquad\quad \sem{\Bunit}\df \{(\bullet,x) \mid x \in \field \}\\
\sem{\Wcomult}\df \left\{\left(x+y,\begin{pmatrix} x\\ {y} \end{pmatrix}\right) \mid x,{y} \in \field \right\}
\qquad\quad
\sem{\Wcounit}\df \{(0,\bullet)\}\qquad\qquad\qquad\quad
\end{gathered}
\end{equation}
\item The image of $\mathsf{P}{\Sigma}$ by $\sem{\cdot}$ is the prop whose arrows $n\to m$ are \emph{linear subspaces} of $\field^n\times\field^m$. We call this prop $\LinRel{\field}$.

\item $\IH$, the theory of \emph{Interacting Hopf algebras} axiomatises linear relations~\cite[Theorem 6.4]{interactinghopf}. It can be found in the second block of \cref{fig:ih-pl}.
\end{itemize}
\item[Affine relations.] \hfill
\begin{itemize}
\item For $\Sigma^+ = \Sigma\cup \{\One\}$ extend $\sem{\cdot}\from \mathsf{P}\Sigma^+ \to \RelX{\field}$ to the prop morphism defined by~\eqref{eq:mat-semantics}-\eqref{eq:linrel-semantics} and
\begin{equation}\label{eq:one-semantics}
\sem{\One} \df \{(\bullet, 1)\}
\end{equation}
\item The image of $\mathsf{P}{\Sigma^+}$ by $\sem{\cdot}$ is the prop whose arrows $n\to m$ are \emph{affine subspaces} of $\field^n\times\field^m$. We call this prop $\AffRel{\field}$.

\item $\IH^+$ adds three simple equations to $\IH$ to obtain an axiomatisation of affine relations~\cite[Theorem 17]{bonchiGraphicalAffineAlgebra2019}. They are the first three axioms of the fourth block of \cref{fig:ih-pl} (the last axiom of this block axiomatises the behaviour of the constant $\One$ in the presence of an order; see immediately below).
\end{itemize}
\item[Polyhedral relations.] We now assume that $\field$ is an ordered field, that is, a field equipped with a total order $\geq$ compatible with the field operations in the following sense: for all $x,y,z\in\field$, \emph{i)} if $x\geq y$ then $x+z\geq y+z$, and \emph{ii)} if $x\geq 0$  and $y\geq 0$ then $xy\geq 0$.
\begin{itemize}
\item For $\Sigma^+_\geq = \Sigma^+\cup \{\scalar{\geq}\}$ extend $\sem{\cdot}\from \mathsf{P}\Sigma^+_{\geq}  \to \RelX{\field}$ to the prop morphism defined by~\eqref{eq:mat-semantics}-\eqref{eq:one-semantics} and
\begin{equation}\label{eq:geq-semantics}
\sem{\scalar{\geq}} \df \{(x,y)\in \field\times \field\mid x \geq y\}
\end{equation}
\item The image of $\mathsf{P}{\Sigma^+_\geq}$ by $\sem{\cdot}$ is the prop whose arrows $n\to m$ are (finitely generated) \emph{polyhedra} of $\field^n\times\field^m$, i.e., subsets of the form 
\[\left\{(x,y)\in \field^n\times\field^m \mid A\begin{pmatrix}x\\ y\end{pmatrix} + b \geq 0 \right\}\] 
for some matrix $A$ and some vector $b$. 
We call this prop $\PolyRel{\field}$ (see \cite{bonchiDiagrammaticPolyhedralAlgebra2021} for more details, in particular Appendix A for the proof that these form a prop).

\item $\IHpoly$ provides an axiomatisation of polyhedral relations~\cite[Corollary 25]{bonchiDiagrammaticPolyhedralAlgebra2021}. It can be found in the first four blocks of \cref{fig:ih-pl}.
\end{itemize}
\end{description}

\subsection{Encoding matrices}

It is helpful to illustrate how matrices are represented as diagrams. An $m\times n$ matrix corresponds to a string diagram $d_A$ with $n$ wires on the left and $m$ wires on the right. The left $j$th port is connected to the $i$th port on the right through an $r$-scalar whenever $A_{ij} = r$. In particular, when $A_{ij}$ is $0$, the two wires are disconnected. For example,
\[A = \begin{pmatrix}r & 0 & 0\\s & 1 & 0\\1 & 0 & 0\\ 0 & 0 & 0\end{pmatrix} \quad \text{ is depicted as } \quad d_A = \;
\tikzset{x=1em, y=2.1ex}
\InputIfFileExists{ex-matrix.tikz}{}{\input{./tikz/ex-matrix.tikz}}
\tikzset{x=1em, y=1.5ex}
\]
One can check that $\sem{d_A} = \{(x,y)\mid y = Ax\}$.

\bibliography{main}

\end{document}